\documentclass[acmsmall,screen=true,review=false]{acmart}

\AtBeginDocument{%
  \providecommand\BibTeX{{%
    \normalfont B\kern-0.5em{\scshape i\kern-0.25em b}\kern-0.8em\TeX}}}

\setcopyright{acmcopyright}
\copyrightyear{2023}
\acmYear{2023}
\acmDOI{XXXXXXX.XXXXXXX}

\acmPrice{15.00}
\acmISBN{978-1-4503-XXXX-X/18/06}

\usepackage{graphicx}
\usepackage{amsfonts}
\usepackage{amsbsy}
\usepackage{bbm}
\usepackage{dsfont}
\usepackage{color}
\usepackage{algorithm2e}
\usepackage{epstopdf}

\newcommand{\Prob}{\mathrm{P}}
\newcommand{\av}{\ensuremath{\mathbf{a}}}
\newcommand{\A}{  \ensuremath{\mathcal{A}}   }
\newcommand{\Ind}{\mathds{1}}
\newcommand{\tv}{\ensuremath{\mathbf{t}}}
\newcommand{\rme}{\mathrm{e}}
\newcommand{\rmd}{\mathrm{d}}
\newcommand{\rmi}{\mathrm{i}}
\newcommand{\R}{\mathbb{R}}
\newcommand{\E}{\mathrm{E}}
\newcommand{\new}{\textcolor{black}}

\begin{document}


\title[How to  generate  a  fault-resilient network  at a lower cost]{How to  generate  a  fault-resilient network  at a lower cost}

\author{Alexander Mozeika}
\email{alexander.mozeika@status.im}
\orcid{0000-0003-1514-1650}
\author{Mohammad M. Jalalzai}
\orcid{0000-0003-0183-4536}
\email{moh@status.im}
\author{Marcin P. Pawlowski}
\orcid{0000-0002-5145-9220}
\email{marcin@status.im}
\affiliation{%
\institution{Status Research \& Development GmbH}
\streetaddress{Baarerstrasse 10}
  \city{Zug}
  \country{Switzerland}
}

\begin{abstract}
Blockchains facilitate decentralization, security, identity, and data management in cyber-physical systems. However,  consensus protocols used in blockchains are prone to high message and computational complexity costs and are not suitable to be used in IoT. One way to reduce message complexity is to randomly assign network nodes into committees or shards. Keeping committee sizes small is then desirable in order to achieve lower message complexity, but this comes with a penalty of reduced reliability as there is a higher probability that a  large number of faulty nodes will end up in a  committee. 
In this work, we study the problem of estimating a probability of a failure in randomly sharded networks. We provide new results and improve existing bounds on the failure probability. Thus,  our framework also paves the way to reduce committee sizes without reducing reliability.

\end{abstract}



\received{20 February 2007}
\received[revised]{12 March 2009}
\received[accepted]{5 June 2009}

\maketitle

\section{Introduction}
Large-scale applications of Cyber-Physical Systems require a significant amount of coordination of a plethora of interconnected devices. Managing thousands of devices pushes the technology to its limits, and creates one of the most important and not yet solved problems of scaling the solution to myriads of connected devices. There are many approaches that try to address this problem efficiently, some of them are focus on centralized solutions, such as cloud   \cite{vasconcelos2019cloud,vzarko2014iot,khodadadi2015data,cao2016marsa,li2017iot}, others are introducing novel architectures like mist and fog \cite{vasconcelos2019cloud,10.1145/2831347.2831354,Kang2018PrivacyPreservedPS,el2021leveraging} or leveraging fully decentralized solutions based on blockchain \cite{michelin2018speedychain}. Each of them has its own merits but in this work, we are focusing on blockchains that have been leveraged to provide decentralized, verifiable, trusted, and traceable IoT-based applications \cite{CommitteeIoTBlockchain}. 

Similarly, due to the strict requirements of IoT networks, blockchains are suitable to defend against data manipulation attacks by providing immutability and avoiding  centralized trusted authority for devices to communicate \cite{SurveyConsensusIoTBlockchain}. Therefore, there have been several research works on leveraging blockchains to provide public safety service \cite{BlockchainPublicSafetyIoT}, smart surveillance for smart cities \cite{BlockchainSmartSurviliance, BlockchainVideoQuery}, automotive industry \cite{fraga2019review} and avionics \cite{BlockchainAvionics}. Furthermore, blockchain technology has also been leveraged to provide    privacy and security in IoT, identity and data management for IoT and monetization for IoT \cite{ApplicationsofBlockchainsinIot}.

Blockchain uses distributed consensus mechanism to reach an agreement for a value to be added to the chain. Consensus makes sure every honest node will eventually agree on the same block at the same sequence in the chain.  Consensus is used to achieve State Machine Replication (SMR).
 It makes sure that each participant of the network maintains the same state all the time. Well-constructed blockchains must be fault-resilient in order to survive Byzantine faults \cite{Lamport:1982:BGP:357172.357176} which are one of the most powerful groups of faults in distributed systems.
A Byzantine Fault Tolerant (BFT) system tolerates arbitrary faults (and adversaries). A BFT SMR system provides SMR services in the presence of Byzantine Faults. The upper bound on tolerating Byzantine faults is $f= \frac{n-1}{3}$ \cite{Fischer:1985:IDC:3149.214121}. This means a system can remain  \emph{reliable} as long as the number of Byzantine faults is less than one-third.

BFT-based consensus protocols provide BFT SMR services. These services constitute the core layer of blockchains. 
Due to the high interest in blockchain, its scalability in terms of the number of nodes (participants) has been one of the main research problems \cite{Luu:2016:SSP:2976749.2978389, Fast-HotStuff,jalalzai2020hermes, Proteus1,2018hotstuff, cypherium, Ethereum-Gasper}. One of the limitations of scaling blockchains is the high message complexity and authenticator (cryptographic signatures) complexity.  Processing a higher number of messages and/or cryptographic operations in a consensus protocol makes it unsuitable for IoT use cases. This is due to the reason that IoT devices have limited bandwidth and processing capacity. Generally, the number of messages as well as cryptographic operations increase as the number of nodes in the network increases.
To scale blockchains (to reduce message and authenticator complexity) different methods were developed to assign nodes randomly in committees \cite{jalalzai2020hermes,Luu:2016:SSP:2976749.2978389, Ethereum-Gasper,AlephZero}. The introduction of committees helps to reduce the message and authenticator complexity from $O(N)$  to $O(n)$ with $n \ll N$, where $n$ is the size of the committee and $N$ is the total number of nodes.

However, randomly assigning nodes into committees is not a straightforward task. One does not want to form a committee with a majority of Byzantine nodes. Therefore, it is paramount to estimate the probability of such an event. Achieving tighter bounds on a committee size for a desired failure probability will allow a protocol to have smaller committees. Smaller committees mean lower message complexity and cryptographic operations, making a protocol more suitable for IoT use cases. 
Therefore, in this work, we present  analytical results  of calculating the probability of a failure of a network in an event of network random partition. First, we introduce  a very general probabilistic framework of random partitions of $N$ nodes of different ``colours''  into $K$ committees.  This framework allows us to construct joint probability distributions of  random partitions. Second, we obtain exact expressions  for these probability distributions  and use them to study the probability of failure.   The latter is the probability of  the event that in at least one of the randomly generated committees, the nodes of some colour exceed a fraction $A$ of all nodes. We consider the probability of failure and we obtain  exact expressions,  bounds, and asymptotic estimates. Finally, we apply our analytical  framework to the sharding of blockchains, which is  a special case of random partitions with two colours  \new{labelling Byzantine and honest nodes}, and improve upon known results in this area. In particular,  we  achieve tighter bounds on committee sizes for a given probability of failure, than the previous known results. Moreover, the term Byzantine and adversarial has been used interchangeably in this document.

\subsection{Related work}
The number of adversaries in a single committee  is usually  modeled with the binomial~\cite{Kokoris2018,Tennakoon2022} and hypergeometric~\cite{Zamani2018, Dang2019, Zhang2022} probability distributions.  The probability of failure is estimated by the union bound which multiplies the probability of failure of  a single committee,  which  uses cumulative distribution function of one of these distributions, by the number of committees $K$. The latter suggests that  these univariate distributions are marginals of some joint probability distribution.  Recently it was established that in hypergeometric case  this joint distribution is the multivariate hypergeometric~\cite{Hafid2020}.  The latter  corresponds to  a sampling scenario where nodes for committees are selected, without replacement, from a  population of $N$ nodes where it is assumed that $M$ nodes in this population are adversarial. However, it is not clear what is this joint distribution for the binomial case. We note that the binomial distribution is quite often chosen as an approximation~\cite{Zamani2018} for the hypergeometric and is  expected to be accurate when the  committee size $n$ is small,  the number of nodes  $N$ is large and the number of adversarial nodes  $M$ is also large.  The probability of failure which  uses \emph{full}  joint probability distribution, to the best of our knowledge, has not been studied  but in  one work~\cite{Hafid2020} where  intuitive derivation of the multivariate hypergeometric distribution was provided  and  corresponding probability of failure  was studied by simulations.  The latter showed advantages of using the joint distribution instead of the union bound  but left analytic solution to be an  open problem~\cite{Hafid2020survey}. The main contribution of this work is to provide a solution to this open problem.

\begin{figure}[t]
\setlength{\unitlength}{1mm}
\begin{center}{
\begin{picture}(100,73)
\put(0,0){\includegraphics[height=75\unitlength]{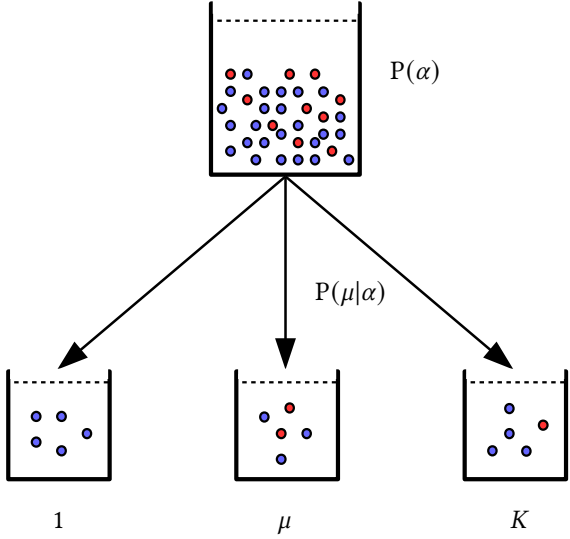}}
\put(57,59){$\Prob(\alpha)$}
\put(47,30){$\Prob(\mu\vert\alpha)$}
\put(12,0){$1$}\put(42,0){$\mu$}\put(73,0){$K$}
\end{picture}
}\end{center}
\caption{The ``urn problem'' associated with different sampling scenarios considered for random partitions. Initially, all $N$ nodes of different colours (here only two colours are shown ) are in one large  committee. The number of nodes of colour $\alpha$ in this committee is  $N\times\Prob(\alpha)$ either on \emph{average} or \emph{exactly}. A node is selected in a random and unbiased way from this large committee  and  is \emph{moved} into one of smaller committees  labeled by  $\mu \in[K]$, where $K<N$. This step is governed by the probability $\Prob(\mu\vert\alpha)$, but possibly subject to \emph{constraints} such as committee sizes and the number of nodes $M=N\times\Prob(\alpha)$ of colour $\alpha$. After $N$ steps \emph{all} nodes will be moved into smaller committees.}
\label{figure:urn-problem} 
\end{figure}
\section{Model of random partitions}
We consider $N$ nodes of $L$ ``colours''  distributed into $K$  shards (or committees).  We assume that partition $\A=(\av_1,\ldots,\av_N)$ with  $\av_i=(a_i(1), a_i(2))$, where\footnote{For $L\in\mathbb{N}$ we use  the definition $[L]=\{1,\ldots,L\}$.}  $a_i(1)\in[L]$ and $a_i(2)\in[K]$, is observed with the probability 
\begin{eqnarray}
\Prob(\A)=\prod_{i=1}^N \Prob(\av_i)  \label{def:P(A)},
\end{eqnarray}
 where $\Prob(\av_i)= \Prob(a_i(1),a_i(2))$
 is the probability that the node $i$ of colour $a_i(1)$ is in the committee  $a_i(2)$. Let us define\footnote{We use the definition   $ \Ind[\mathcal{E}]\in\{0,1\}$ for  the indicator function of some event $\mathcal{E}$ and $\delta_{x;y}$  is  the Kronecker delta function, i.e.  $\delta_{x;y}=\Ind[x=y]$ for   $x,y\in\mathbb{Z}$.} the random variables    $N_{\mu}(\A)= \sum_{i=1}^N\delta_{\mu;\, a_i(2)}$, i.e. the total number of nodes  in committee $\mu$, $N_{\mu}^\alpha(\A)= \sum_{i=1}^N\delta_{\alpha;\, a_i(1)}\delta_{\mu;\, a_i(2)}$, i.e. the number of nodes of colour $\alpha$ in committee $\mu$, and $M= \sum_{\mu=1}^K N_{\mu}^\alpha(\A)$, i.e. the  total number of nodes of colour $\alpha$. We note that $\sum_{\mu=1}^K N_{\mu}(\A)=N$ and $N_{\mu}(\A)= \sum_{\alpha=1}^L N_{\mu}^\alpha(\A)$. The joint probability distribution of these random variables in all of the $K$ committees is given by 
\begin{eqnarray}
\Prob\left(N_{1}^\alpha,\ldots, N_{K}^\alpha; N_{1},\ldots,N_{K};M\right)&=&\sum_{\A}\Prob(\A)\,\delta_{M;\sum_{\mu=1}^K N_{\mu}^\alpha(\A)}\prod_{\mu=1}^K\delta_{N_{\mu}^\alpha;N_{\mu}^\alpha(\A)}\delta_{N_{\mu};N_{\mu}(\A)}\label{def:Prob-joint-master},
\end{eqnarray}
where $\sum_{\mu=1}^KN_{\mu}=N$,  $N_{\mu}^\alpha\leq N_{\mu}$ for all $\mu$ and $M\leq N$. 

From  the probability distribution (\ref{def:Prob-joint-master})  we construct, by  marginalization  and application of Bayes' theorem, the joint distribution 
\begin{eqnarray}
%
\Prob\left(N_{1}^\alpha,\ldots, N_{K}^\alpha; N_{1},\ldots,N_{K}\right)&=&\sum_{\A}\Prob(\A)\prod_{\mu=1}^K\delta_{N_{\mu}^\alpha;N_{\mu}^\alpha(\A)}\delta_{N_{\mu};N_{\mu}(\A)}\label{def:Prob}
\end{eqnarray}
the conditional distribution 
\begin{eqnarray}
\Prob\left(N_{1}^\alpha,\ldots, N_{K}^\alpha\vert N_{1},\ldots,N_{K}\right)
&=&\frac{\Prob\left(N_{1}^\alpha,\ldots, N_{K}^\alpha; N_{1},\ldots,N_{K}\right)}{\Prob\left(N_{1},\ldots,N_{K}\right)}\label{def:Prob-cond} 
\end{eqnarray}
and the conditional distribution
\begin{eqnarray}
\Prob\left(N_{1}^\alpha,\ldots, N_{K}^\alpha\vert N_{1},\ldots,N_{K}; M\right)
&=&\frac{\Prob\left(N_{1}^\alpha,\ldots, N_{K}^\alpha; N_{1},\ldots,N_{K};M\right)}{\Prob\left(N_{1},\ldots,N_{K};M\right)}\label{def:Prob-cond-M}. 
\end{eqnarray}
 These  probability distributions  correspond to three different sampling scenarios  of random partitions (see Figure \ref{figure:urn-problem}).  In the first  scenario, described by (\ref{def:Prob}), each node is assigned randomly to a committee independently from other nodes. The second scenario, described by (\ref{def:Prob-cond}),  corresponds to fixed numbers of (randomly selected) nodes being  assigned  to committees. Finally, the third scenario, described by (\ref{def:Prob-cond-M}), is similar to the second one  but when we have \emph{exactly} $M$ nodes of colour $\alpha$.

\section{Statistical properties of  random partitions\label{section:random-part}}

\subsection{Probability distributions\label{ssection:distr}}
For the   probability distributions (\ref{def:Prob})-(\ref{def:Prob-cond-M}) we obtain the following results.
\begin{lemma}
\label{lemma:joint-distr}
Assuming that partition $\A$ is generated with the probability   (\ref{def:P(A)}) the joint distribution (\ref{def:Prob}) is the product 
\begin{eqnarray}
\Prob\left(N_1,N^\alpha_1,\cdots,N_K,N^\alpha_K \right)
&=&\Prob\left(N_1,\cdots,N_K\vert N \right)\prod_{\mu=1}^K\Prob^\alpha_\mu\left(N^\alpha_\mu\vert N_\mu \right)\label{eq:Prob-joint}
\end{eqnarray}
of the multinomial  
\begin{eqnarray}
\Prob\left(N_1,\cdots,N_K\vert N \right)&=&\Ind\left[\sum_{\mu=1}^KN_\mu=N\right]\frac{N!}{\prod_{\mu=1}^KN_\mu!}\prod_{\mu=1}^K \Prob^{N_\mu}(\mu)\label{def:multinomial-dist}
\end{eqnarray}
and the product  $\prod_{\mu=1}^K\Prob^\alpha_\mu\left(N^\alpha_\mu\vert N_\mu \right)$ of the binomial
\begin{eqnarray}
\Prob^\alpha_\mu\left(N^\alpha_\mu\vert N_\mu \right)&=&{N_\mu\choose N^\alpha_\mu}\,\Prob^{N^\alpha_\mu}(\alpha\vert \mu)  \left[1-\Prob(\alpha\vert \mu)\right]^{N_\mu-N^\alpha_\mu}
\label{def:binomial-distr}
\end{eqnarray}
distributions.
\end{lemma}

\begin{proof}
To show this we consider the moment-generating function   
\begin{eqnarray}
\mathrm{Z}\left[\tv_1,\tv_2\right]&=&\sum_{N_{1}}\sum_{N_{1}^\alpha}\cdots\sum_{N_{K}}\sum_{N_{K}^\alpha}
\,\Prob\left(N_{1}^\alpha,N_{1},\ldots, N_{K}^\alpha,N_{K}\right)\nonumber\\
&&~~~~~~~~~~\times\rme^{\sum_{\mu=1}^K\left\{N_{\mu}^\alpha\, t_1(\mu)+N_{\mu}\,t_2(\mu)\right\}}\label{def:MGF},
\end{eqnarray}
where $\tv_1,\tv_2\in\R^K$. Now using in above the definition (\ref{def:Prob}) we obtain 
\begin{eqnarray}
\mathrm{Z}\left[\tv_1,\tv_2\right]&=&\sum_{\A}\Prob(\A)\,\rme^{\sum_{\mu=1}^K\left\{N_{\mu}^\alpha(\A)\, t_1(\mu)+N_{\mu}(\A)\,t_2(\mu)\right\}}\nonumber\\
&=&\left[\sum_{\beta=1}^L \sum_{\nu=1}^K\Prob(\beta, \nu)\,\rme^{\sum_{\mu=1}^K\,\delta_{\mu;\, \nu}\{t_1(\mu)\,\delta_{\alpha;\, \beta} +t_2(\mu)\}}\right]^N\nonumber\\
&=&\sum_{N_1=0}^N\cdots\sum_{N_K=0}^N \Prob\left(N_1,\cdots,N_K\vert N \right)\prod_{\mu=1}^K\Prob^\alpha_\mu\left(N^\alpha_\mu\vert N_\mu \right)\nonumber\\
&&\times \rme^{\sum_{\mu=1}^K\{t_1(\mu)\,N^\alpha_\mu +t_2(\mu)\,N_\mu\}}\label{eq:MGF}
\end{eqnarray}
and hence comparing above with  (\ref{def:MGF})  we arrive at  (\ref{eq:Prob-joint}).
\end{proof}
The consequence of the \emph{Lemma} \ref{lemma:joint-distr} is that the conditional probability distribution  (\ref{def:Prob-cond}) is given  by 
\begin{eqnarray}
\Prob\left(N_{1}^\alpha,\ldots, N_{K}^\alpha\vert N_{1},\ldots,N_{K}\right)&=&\prod_{\mu=1}^K\Prob^\alpha_\mu\left(N^\alpha_\mu\vert N_\mu \right)\label{eq:Prob-cond},
\end{eqnarray}
where $\sum_{\mu=1}^KN_\mu=N$, which follows from the equation  (\ref{eq:Prob-joint}) and from the  property  $\Prob(Y,X)=\Prob(Y\vert X)\,\Prob(X)$ of a  joint probability distribution. Furthermore,  the marginal of (\ref{def:Prob}) is the product of distributions 
\begin{eqnarray}
\Prob^\alpha\!\left(N_\mu,N^\alpha_\mu\vert N\right)
&=&\Prob_\mu\!\left(N_\mu\vert N\right)\,\Prob^\alpha_\mu\!\left(N^\alpha_\mu\vert N_\mu \right),\label{eq:P-marginal}
\end{eqnarray}
where  $\Prob_\mu\left(N_\mu\vert N\right)$ is the binomial distribution 
\begin{eqnarray}
\Prob_\mu\left(N_\mu\vert N \right)&=&{N\choose N_\mu}\,\Prob^{N_\mu}(\mu)  \left[1-\Prob(\mu)\right]^{N-N_\mu}.
\end{eqnarray}
Finally, we  also obtain  the following 
\begin{corollary}
 The conditional probability (\ref{def:Prob-cond-M}) is the multivariate hypergeometric  distribution 
\begin{eqnarray}
%
\Prob\left(N_{1}^\alpha,\ldots, N_{K}^\alpha\vert N_{1},\ldots,N_{K}; M\right)
&=&\frac{\delta_{M;\sum_{\mu=1}^K N_{\mu}^\alpha}\prod_{\mu=1}^K {N_\mu\choose N^\alpha_\mu}}{{N\choose M}}\label{eq:Prob-cond-M}.
\end{eqnarray}
\end{corollary}

\begin{proof}
To show this  we first rewrite the probability distribution (\ref{def:Prob-cond-M}) as follows
\begin{eqnarray}
\Prob\left(N_{1}^\alpha,\ldots, N_{K}^\alpha\vert N_{1},\ldots,N_{K}; M\right)
&=&\frac{\delta_{M;\sum_{\mu=1}^K N_{\mu}^\alpha}\Prob\left(N_{1}^\alpha,N_{1},\ldots, N_{K}^\alpha, N_{K}\right)}{\sum_{\tilde{N}_{1}^\alpha} \cdots \sum_{\tilde{N}_{1}^\alpha}\delta_{M;\sum_{\mu=1}^K \tilde{N}_{\mu}^\alpha}\Prob\left(\tilde{N}_{1}^\alpha,N_{1},\ldots, \tilde{N}_{K}^\alpha, N_{K}\right)}\nonumber\\
&=&\frac{\delta_{M;\sum_{\mu=1}^K N_{\mu}^\alpha}\Prob\left(N_{1}^\alpha,\ldots, N_{K}^\alpha\vert N_{1},\ldots,N_{K}\right)}{\sum_{\tilde{N}_{1}^\alpha} \cdots \sum_{\tilde{N}_{1}^\alpha}\delta_{M;\sum_{\mu=1}^K \tilde{N}_{\mu}^\alpha}\Prob\left(\tilde{N}_{1}^\alpha,\ldots, \tilde{N}_{K}^\alpha\vert N_{1},\ldots,N_{K}\right)}.
\label{eq:Prob-cond-M-der-1} 
\end{eqnarray}
Second, by equation  (\ref{eq:Prob-cond}) the numerator
\begin{eqnarray}
\delta_{M;\sum_{\mu=1}^K N_{\mu}^\alpha}\Prob\left(N_{1}^\alpha,\ldots, N_{K}^\alpha\vert N_{1},\ldots,N_{K}\right)
&=&\delta_{M;\sum_{\mu=1}^K N_{\mu}^\alpha}\prod_{\mu=1}^K\Prob^\alpha_\mu\left(N^\alpha_\mu\vert N_\mu \right)\nonumber\\
%
%
&=&\delta_{M;\sum_{\mu=1}^K N_{\mu}^\alpha}\Prob^{M}(\alpha\vert \mu)  \left[1-\Prob(\alpha\vert \mu)\right]^{N-M}\prod_{\mu=1}^K {N_\mu\choose N^\alpha_\mu}
\end{eqnarray}
and hence  
\begin{eqnarray}
\Prob\left(N_{1}^\alpha,\ldots, N_{K}^\alpha\vert N_{1},\ldots,N_{K}; M\right)&=&\frac{\delta_{M;\sum_{\mu=1}^K N_{\mu}^\alpha}\prod_{\mu=1}^K {N_\mu\choose N^\alpha_\mu}}{\sum_{\tilde{N}^\alpha_1=0}^{N_1}\cdots\sum_{\tilde{N}^\alpha_K=0}^{N_K}\delta_{M;\sum_{\mu=1}^K \tilde{N}_{\mu}^\alpha}\prod_{\mu=1}^K {N_\mu\choose \tilde{N}^\alpha_\mu}}\label{eq:Prob-cond-M-der-2}.
\end{eqnarray}
Now the denominator, using the integral representation $\delta_{k;n}=\int_{-\pi}^{\pi} \frac{\rmd \hat{m}}{2\pi}\rme^{\rmi\hat{m}(k-n)}$, is given by
\begin{eqnarray}
\sum_{N^\alpha_1=0}^{N_1}\cdots\sum_{N^\alpha_K=0}^{N_K}\delta_{M;\sum_{\mu=1}^K N_{\mu}^\alpha}\prod_{\mu=1}^K {N_\mu\choose N^\alpha_\mu}&=&\int_{-\pi}^{\pi} \frac{\rmd \hat{m}}{2\pi}\rme^{\rmi\hat{m}M}\prod_{\mu=1}^K\sum_{N^\alpha_\mu=0}^{N_\mu} {N_\mu\choose N^\alpha_\mu}\rme^{-\rmi\hat{m} N_{\mu}^\alpha}\nonumber\\
%
%
&=&\int_{-\pi}^{\pi} \frac{\rmd \hat{m}}{2\pi}\rme^{\rmi\hat{m}M}\left(1+\rme^{-\rmi \hat{m}}\right)^{\sum_{\mu=1}^KN_\mu}={N\choose M}\label{eq:hyper-norm}
%
%
\end{eqnarray}
which completes the proof. 
\end{proof}
We note that  marginal of the hypergeometric distribution (\ref{eq:Prob-cond-M}) can be derived in a similar way as the sum in   (\ref{eq:hyper-norm}). In particular,  we consider, without loss of generality, the sum 
\begin{eqnarray}
\sum_{N^\alpha_2=0}^{N_2}\cdots\sum_{N^\alpha_K=0}^{N_K}\delta_{M;\sum_{\mu=1}^K N_{\mu}^\alpha}\prod_{\mu=1}^K {N_\mu\choose N^\alpha_\mu}
&=&\int_{-\pi}^{\pi} \frac{\rmd \hat{m}}{2\pi}\rme^{\rmi\hat{m}M}{N_1\choose N^\alpha_1}\rme^{-\rmi\hat{m} N_{1}^\alpha}\prod_{\mu=2}^K\sum_{N^\alpha_\mu=0}^{N_\mu} {N_\mu\choose N^\alpha_\mu}\rme^{-\rmi\hat{m} N_{\mu}^\alpha}\nonumber\\
&&={N_1\choose N^\alpha_1}\int_{-\pi}^{\pi} \frac{\rmd \hat{m}}{2\pi}\rme^{\rmi\hat{m}(M-N_{1}^\alpha)}\left(1+\rme^{-\rmi \hat{m}}\right)^{\sum_{\mu=2}^KN_\mu}\nonumber\\
&&={N_1\choose N^\alpha_1}{N-N_1\choose M-N^\alpha_1}
\end{eqnarray}
and hence the marginal of  (\ref{eq:Prob-cond-M}) is given by
\begin{eqnarray}
\Prob\left(N_{\mu}^\alpha\vert N_{\mu}; M\right)&=&\frac{{N_\mu\choose N^\alpha_\mu}{N-N_\mu\choose M-N^\alpha_\mu}}{{N\choose M}}=\frac{{M\choose N^\alpha_\mu}{N-M\choose N_\mu-N^\alpha_\mu}}{{N\choose N_\mu}}\label{eq:hypergeom-univ}.
\end{eqnarray}

\subsection{The probability of failure\label{ssection:failure}}
In random partitions sampled from the probability distributions (\ref{def:Prob})-(\ref{def:Prob-cond-M})  we are interested in   the event 
\begin{eqnarray}
%
\E&=&\sum_{\mu=1}^K\Ind\left[N_{\mu}^\alpha\geq \lfloor A N_{\mu}\rfloor+1\right]>0
\label{def:E}, 
\end{eqnarray}
 i.e.  in at  least \emph{one} committee the nodes of colour $\alpha$ exceed a fraction $A\in(0,1)$ of all nodes in the committee. We will call such event  \emph{failure} and define the probability of failure  
\begin{eqnarray}
%
\delta&=&\Prob\left(\E>0\right)\nonumber\\
&=&1-\Prob\left(\E=0\right)\label{def:delta}. 
\end{eqnarray}
 We note that  event $\E=0$ is equivalent to the $\prod_{\mu=1}^K\Ind\left[N_{\mu}^\alpha\leq \lfloor AN_\mu\rfloor\right]=1$ and for  the probability of failure in random partitions  we obtain the following results. 
\begin{theorem}
\label{theorem:lb-ub}
  For random partitions  sampled from the probability distribution (\ref{def:Prob-cond}) when  $\Prob(\alpha\vert\mu)<Q(\mu)<1$, where $Q(\mu)=\frac{\lfloor A N_\mu\rfloor+1}{N_\mu}$,  the probability of failure  $\delta$ is bounded as follows 
\begin{eqnarray}
1-\prod_{\mu=1}^K\left[1-\frac{\rme^{-N_\mu \mathrm{D}\left(Q(\mu)\vert\vert\Prob(\alpha\vert\mu)\right)}}{\sqrt{8N_\mu Q(\mu)\left(1-Q(\mu)\right)}}\right]\leq\delta\leq 1-\prod_{\mu=1}^K\left[1-\rme^{-N_\mu \mathrm{D}\left(Q(\mu)\vert\vert\Prob(\alpha\vert\mu)\right)}\right]\label{eq:lb-ub},
\end{eqnarray}
where $\mathrm{D}\left(Q\vert\vert P\right)$ with  $Q,P\in(0,1)$ is the Kullback$-$Leibler  (KL) divergence
\begin{eqnarray}
\mathrm{D}\left(Q\vert\vert P\right)&=&Q\log\frac{Q}{P}+(1-Q)\log\frac{1-Q}{1-P}
\end{eqnarray}
which is $0$  when $Q=P$ and is positive semi-definite when $Q\neq P$~\cite{Cover2012}.
\end{theorem}
\begin{proof}
To show this we first, using the probability distribution (\ref{eq:Prob-cond}), define the probability\footnote{We use the definitions $\left\{N_{\mu}^\alpha\leq \lfloor AN_\mu\rfloor\right\}\equiv N_{1}^\alpha\leq \lfloor AN_1\rfloor,\ldots,N_{K}^\alpha\leq \lfloor AN_K\rfloor$ and $\{N_{\mu}\}\equiv N_{1},\ldots, N_{K}$.} 
\begin{eqnarray}
\Prob\left(\left\{N_{\mu}^\alpha\leq \lfloor AN_\mu\rfloor\right\}\vert \left\{N_{\mu}\right\}\right)&=&\sum_{N_{1}^\alpha}\cdots\sum_{N_{K}^\alpha}\Prob\left(N_{1}^\alpha,\ldots, N_{K}^\alpha\vert N_{1},\ldots,N_{K}\right)\prod_{\mu=1}^K\Ind\left[N_{\mu}^\alpha\leq \lfloor AN_\mu\rfloor\right]
\end{eqnarray}
and consider the probability of failure   
\begin{eqnarray}
\delta=1- \Prob\left(\left\{N_{\mu}^\alpha\leq \lfloor AN_\mu\rfloor\right\}\vert \left\{N_{\mu}\right\}\right)&=&1-\prod_{\mu=1}^K\sum_{N^\alpha_\mu=0}^{\lfloor A N_\mu\rfloor}\Prob^\alpha_\mu\left(N^\alpha_\mu\vert N_\mu \right)\nonumber\\
&=&1-\prod_{\mu=1}^K\left[1-\Prob(N^\alpha\geq \lfloor A N_\mu\rfloor+1\vert N_\mu)\right]\label{eq:delta-binom},
\end{eqnarray}
where we defined the probability
\begin{eqnarray}
\Prob(N^\alpha\geq \lfloor A N_\mu\rfloor+1\vert N_\mu)&=&\sum_{N^\alpha=\lfloor A N_\mu\rfloor+1}^{N_\mu}\Prob^\alpha_\mu\left(N^\alpha\vert N_\mu \right). 
\end{eqnarray}
Second, the probability distribution   $\Prob^\alpha_\mu\left(N^\alpha\vert N_\mu \right)$, which was defined in (\ref{def:binomial-distr}),  is the binomial    and hence it   can be bounded from the above   
\begin{eqnarray}
&&\Prob(N^\alpha\geq \lfloor A N_\mu\rfloor+1\vert N_\mu)\leq \rme^{-N_\mu \mathrm{D}\left(Q(\mu)\vert\vert\Prob(\alpha\vert\mu)\right)}\label{eq:binom-tail-ub-1}
\end{eqnarray}
and from the below 
\begin{eqnarray}
&&\Prob(N^\alpha\geq \lfloor A N_\mu\rfloor+1\vert N_\mu)
\geq \frac{\rme^{-N_\mu \mathrm{D}\left(Q(\mu)\vert\vert\Prob(\alpha\vert\mu)\right)}}{\sqrt{8N_\mu Q(\mu)\left(1-Q(\mu)\right)}}
\end{eqnarray}
for $\Prob(\alpha\vert\mu)<Q(\mu)<1$, where $Q(\mu)=\frac{\lfloor A N_\mu\rfloor+1}{N_\mu}$,
by the \emph{Lemma} 4.7.2 in \cite{Ash2012}. Using these bounds in (\ref{eq:delta-binom}) gives us the lower and upper bound in (\ref{eq:lb-ub}).
\end{proof}
We note that a \emph{tighter} upper bound in (\ref{eq:lb-ub})  can be obtained by using the recent result 
\begin{eqnarray}
\Prob(N^\alpha\geq \lfloor A N_\mu\rfloor+1\vert N_\mu)\leq \frac{1}{1-r(\mu)}\frac{\rme^{-N_\mu \mathrm{D}\left(Q(\mu)\vert\vert\Prob(\alpha\vert\mu)\right)}}{\sqrt{2\pi Q(\mu)\left(1-Q(\mu)\right)N_\mu}},\label{eq:binom-tail-ub-2}
\end{eqnarray}
where $\Prob(\alpha\vert\mu)<Q(\mu)<1$, of  the \emph{Theorem} 1 in~\cite{Ferrante2021}. Here the term $\frac{1}{1-r(\mu)}$, where $r(\mu)=\frac{\Prob(\alpha\vert\mu)\left(1-Q(\mu)\right)}{ Q(\mu)\left(1-\Prob(\alpha\vert\mu)\right)}$, is an  upper bound on the sum  $\sum_{k=0}^{\left(1-Q(\mu)\right)N_\mu}r^k(\mu)=(1-r^{\left(1-Q(\mu)\right)N_\mu+1}(\mu))/(1-r(\mu))$. 

\subsubsection{Saddle-point method}
For random partitions generated by the probability distribution (\ref{def:Prob-cond-M}), i.e. \emph{exactly} $M$  nodes out of $N$ are of colour $\alpha$, the probability of failure is given by  
\begin{eqnarray}
\delta&=&1-\Prob\left(\{N_{\mu}^\alpha\leq \lfloor AN_\mu\rfloor\}\vert \{N_\mu\}, M\right)\label{eq:delta-hyper}, 
\end{eqnarray}
where the probability 
\begin{eqnarray}
\Prob\left(\{N_{\mu}^\alpha\leq \lfloor AN_\mu\rfloor\}\vert \{N_{\mu}\};M\right)
&=&\sum_{N^\alpha_1=0}^{N_1}\cdots\sum_{N^\alpha_K=0}^{N_K}\Prob\left(N_{1}^\alpha,\ldots, N_{K}^\alpha\vert N_{1},\ldots,N_{K}; M\right)\nonumber\\
&&\times\prod_{\mu=1}^K\Ind\left[N_{\mu}^\alpha\leq \lfloor AN_\mu\rfloor\right]\nonumber\\
&=&\frac{\sum_{N^\alpha_1=0}^{N_1}\cdots\sum_{N^\alpha_K=0}^{N_K}\delta_{M;\sum_{\mu=1}^K N_{\mu}^\alpha}\prod_{\mu=1}^K {N_\mu\choose N^\alpha_\mu}\Ind\left[N_{\mu}^\alpha\leq \lfloor AN_\mu\rfloor\right]}{{N\choose M}}\label{eq:1-delta}
\end{eqnarray}
follows from  the result (\ref{eq:Prob-cond-M}) for the distribution (\ref{def:Prob-cond-M}). 

We would like to obtain  a simpler  analytic expression for  $\delta$ but the challenge here is  in computing  the sums  in  (\ref{eq:1-delta}) efficiently. It is not clear how to do this for any $N$, $M$ and $K$, but in the limit $N\rightarrow\infty$ with  $N_\mu<\infty$ such that $N=\sum_{\mu=1}^K N_\mu$ and $M/N\in (0,1)$ we can obtain  the following result
\begin{eqnarray}
\Prob\left(\{N_{\mu}^\alpha\leq \lfloor AN_\mu\rfloor\}\vert \{N_\mu\}, M\right)&=&\!
\sqrt{\frac{ 
N\,P(1-P)
}  {\sum_{\mu=1}^K\left[\langle(N^\alpha)^2\rangle_{A,Q,N_\mu;}-\langle N^\alpha\rangle^2_{A,Q,N_\mu;}\right]  }}\rme^{N\Psi[Q]}\!+\!O(1/N)\label{eq:SP-Prob-N-large-Q},
\end{eqnarray}
where $P=M/N$ and we defined the function 
\begin{eqnarray}
\Psi[Q]&=&D(P\vert\vert Q)+\frac{1}{N}\sum_{\mu=1}^K\log\sum_{N^\alpha=0}^{\lfloor AN_\mu\rfloor} {N_\mu\choose N^\alpha}\,Q^{N^\alpha}[1-Q]^{N_\mu-N^\alpha}\label{def:Psi-Q}
\end{eqnarray}
and average 
\begin{eqnarray}
\langle f(N^\alpha)\rangle_{A,Q,N_\mu;}&=&\frac{\sum_{N^\alpha=0}^{\lfloor AN_\mu\rfloor} {N_\mu\choose N^\alpha}Q^{N^\alpha}[1-Q]^{N_\mu-N^\alpha}f(N^\alpha)}{\sum_{N^\alpha=0}^{\lfloor AN_\mu\rfloor} {N_\mu\choose N^\alpha}Q^{N^\alpha}[1-Q]^{N_\mu-N^\alpha}}
\end{eqnarray}
for any function $f(N^\alpha)$.  In above  $Q\in (0,1)$  is the solution of the equation 
\begin{eqnarray}
P&=& \frac{1}{N}\sum_{\mu=1}^K
\langle N^\alpha\rangle_{A,Q,N_\mu;}\label{eq:SP-large-N-Q}
\end{eqnarray}
for  $P\leq A$.

In this work,  we will  present only a heuristic  argument that leads to the equation   (\ref{eq:SP-Prob-N-large-Q}), but we envisage that a more rigorous proof of the latter is also possible. Our approach to computing (\ref{eq:1-delta}) is to use the \emph{saddle-point} method of integration, which is quite often used in  statistical physics~\cite{Nishimori2001} and analytic combinatorics \cite{Flajolet2009} to compute moment generating functions,  but first in order to apply  this method we need to represent the sums in  (\ref{eq:1-delta}) as an integral.  To this end, we rewrite the numerator in (\ref{eq:1-delta}) as follows 
\begin{eqnarray}
&&\sum_{N^\alpha_1=0}^{N_1}\cdots\sum_{N^\alpha_K=0}^{N_K}\delta_{M;\sum_{\mu=1}^K N_{\mu}^\alpha}\prod_{\mu=1}^K {N_\mu\choose N^\alpha_\mu}\Ind\left[N_{\mu}^\alpha\leq \lfloor AN_\mu\rfloor\right]\nonumber\\
&&~~~~~=\sum_{N^\alpha_1=0}^{N_1}\cdots\sum_{N^\alpha_K=0}^{N_K} 
\frac{1}{2\pi\rmi}\oint_{\vert z\vert=1}\!\! z^{-M-1+\sum_{\mu=1}^KN_{\mu}^\alpha}\,\rmd z\prod_{\mu=1}^K {N_\mu\choose N^\alpha_\mu}\Ind\left[N_{\mu}^\alpha\leq \lfloor AN_\mu\rfloor\right]\nonumber\\
&&~~~=\frac{1}{2\pi\rmi}\oint_{\vert z\vert=1}\!\left\{  
\prod_{\mu=1}^K\sum_{N^\alpha=0}^{N_\mu} {N_\mu\choose N^\alpha}\Ind\left[N^\alpha\leq \lfloor AN_\mu\rfloor\right]z^{N^\alpha}\right\}z^{-M-1}\rmd z\nonumber\\
&&~~~=\frac{1}{2\pi\rmi}\oint_{\vert z\vert=1}\! 
\left\{\prod_{\mu=1}^K\phi_A(z\vert N_\mu)\right\}\,z^{-M-1}\rmd z,
%
%
\end{eqnarray}
where we used the integral representation   $\delta_{n;m}=\frac{1}{2\pi\rmi}\oint_{\vert z\vert=1} z^{-n-1+m}\,\rmd z$ and defined the function 
\begin{eqnarray}
\phi_{A}(z\vert N_\mu)=\sum_{N^\alpha=0}^{\lfloor AN_\mu\rfloor} {N_\mu\choose N^\alpha}z^{N^\alpha}\label{def:phi_A}.
\end{eqnarray}

Second, the denominator
\begin{eqnarray}
{N\choose M}&=&\sum_{N^\alpha_1=0}^{N_1}\cdots\sum_{N^\alpha_K=0}^{N_K}\delta_{M;\sum_{\mu=1}^K N_{\mu}^\alpha}\prod_{\mu=1}^K {N_\mu\choose N^\alpha_\mu}\nonumber\\
&=&\sum_{N^\alpha_1=0}^{N_1}\cdots\sum_{N^\alpha_K=0}^{N_K} 
\frac{1}{2\pi\rmi}\oint_{\vert z\vert=1}\!\! z^{-M-1+\sum_{\mu=1}^KN_{\mu}^\alpha}\,\rmd z\prod_{\mu=1}^K {N_\mu\choose N^\alpha_\mu}\nonumber\\
&=&\frac{1}{2\pi\rmi}\oint_{\vert z\vert=1}\! 
\left\{\prod_{\mu=1}^K\phi_1(z\vert N_\mu)\right\}\,z^{-M-1}\rmd z
%
%
\end{eqnarray}
and hence 
\begin{eqnarray}
\Prob\left(\{N_{\mu}^\alpha\leq \lfloor AN_\mu\rfloor\}\vert \{N_\mu\}, M\right)&=&\frac{\oint_{\vert z\vert=1}\! 
\left\{\prod_{\mu=1}^K\phi_A(z\vert N_\mu)\right\}\,z^{-M-1}\rmd z}{\oint_{\vert z\vert=1}\! 
\left\{\prod_{\mu=1}^K\phi_1(z\vert N_\mu)\right\}\,z^{-M-1}\rmd z}.
\end{eqnarray}

Let us  consider the scenario of $N\rightarrow\infty$ with  $N_\mu<\infty$ such that $N=\sum_{\mu=1}^K N_\mu$ and $0<M/N<1$. Then for the integral
\begin{eqnarray}
\oint_{\vert z\vert=1}\! 
\left\{\prod_{\mu=1}^K\phi_A(z\vert N_\mu)\right\}\,z^{-M-1}\rmd z
&=&\oint_{\vert z\vert=1}\rmd z\,\rme^{N\Psi_{A}(z)-\log(z)},
\end{eqnarray}
where in above we defined  $P=M/N$ and 
\begin{eqnarray}
\Psi_{A}(z)&=&-P\log(z) +\frac{1}{N}\sum_{\mu=1}^K\log\phi_A(z\vert N_\mu),
\end{eqnarray}
we can  try to use the saddle-point integration method~\cite{Fedoryuk1977}. Applying this method allows us to write 
\begin{eqnarray}
\oint_{\vert z\vert=1}\!\!\rmd z\,\rme^{N\Psi_{A}(z)-\log(z)}
&=&\sqrt{-\frac{2\pi}{N\ddot{\Psi}_{A}(z_0(A))}}\left[1/z_0(A)+O(1/N)\right]\rme^{N\Psi_{A}(z_0(A))},
\end{eqnarray}
where in above we used the definition $\ddot{\Psi}_A(z)=\frac{\partial^2}{\partial z^2} \Psi_A(z)$ and $z_0(A)$ is the solution of the equation $\dot{\Psi}_A(z)=0$, which is given by 
\begin{eqnarray}
P&=& \frac{1}{N}\sum_{\mu=1}^K
\frac{\sum_{N^\alpha=0}^{\lfloor AN_\mu\rfloor} {N_\mu\choose N^\alpha}z^{N^\alpha}N^\alpha}
{\sum_{N^\alpha=0}^{\lfloor AN_\mu\rfloor} {N_\mu\choose N^\alpha}z^{N^\alpha}}\label{eq:SP-large-N}.
\end{eqnarray}
Thus, for the probability (\ref{eq:1-delta}) we obtain the following expression 
\begin{eqnarray}
\Prob\left(\{N_{\mu}^\alpha\leq \lfloor AN_\mu\rfloor\}\vert \{N_\mu\}, M\right)
&=&\frac{1/z_0(A)+O(1/N)}{1/z_0(1)+O(1/N)}
\sqrt{\frac{\ddot{\Psi}_{1}(z_0(1))}{\ddot{\Psi}_{A}(z_0(A))}}
\rme^{N[\Psi_{A}(z_0(A))-\Psi_{1}(z_0(1))]}\label{eq:1-delta-SP-large-N}.
\end{eqnarray}

Let us now consider the equation (\ref{eq:SP-large-N}). In the latter $0<P<1$ and hence for $z>0$ we can set $z=Q/(1-Q)$, where $0<Q<1$, in this equation, giving us  
\begin{eqnarray}
P&=& \frac{1}{N}\sum_{\mu=1}^K
\langle N^\alpha\rangle_{A,Q,N_\mu;}
\end{eqnarray}
where we defined the average 
\begin{eqnarray}
\langle f(N^\alpha)\rangle_{A,Q,N_\mu;}&=&\frac{\sum_{N^\alpha=0}^{\lfloor AN_\mu\rfloor} {N_\mu\choose N^\alpha}Q^{N^\alpha}[1-Q]^{N_\mu-N^\alpha}f(N^\alpha)}{\sum_{N^\alpha=0}^{\lfloor AN_\mu\rfloor} {N_\mu\choose N^\alpha}Q^{N^\alpha}[1-Q]^{N_\mu-N^\alpha}}.
\end{eqnarray}
Using  $N^\alpha\leq \lfloor AN_\mu\rfloor$ in the equation (\ref{eq:SP-large-N-Q})  gives us the inequality  
\begin{eqnarray}
P&\leq& \frac{1}{N}\sum_{\mu=1}^K
\frac{\sum_{N^\alpha=0}^{\lfloor AN_\mu\rfloor} {N_\mu\choose N^\alpha}\,Q^{N^\alpha}[1-Q]^{N_\mu-N^\alpha}\lfloor AN_\mu\rfloor}
{\sum_{N^\alpha=0}^{\lfloor AN_\mu\rfloor} {N_\mu\choose N^\alpha}Q^{N^\alpha}[1-Q]^{N_\mu-N^\alpha}}\nonumber\\
&&~~~~~~~~~~=\frac{1}{N}\sum_{\mu=1}^K\lfloor AN_\mu\rfloor\leq A\label{eq:SP-large-N-Q-ineq}
\end{eqnarray}
and hence equation (\ref{eq:SP-large-N-Q}) has a solution only for $P\leq A$. 

Furthermore, for $A=1$  the $Q=P$ is the solution of the equation (\ref{eq:SP-large-N-Q})  giving us  $z_0(1)=P/(1-P)$. The  latter can be used to compute  the function
\begin{eqnarray}
\Psi_{1}(z_0(1))&=&-P\log(P/(1-P))-\log(1-P)+\frac{1}{N}\sum_{\mu=1}^K\log\sum_{N^\alpha=0}^{N_\mu} {N_\mu\choose N^\alpha}P^{N^\alpha}[1-P]^{N_\mu-N^\alpha}\nonumber\\
&=&\mathcal{S}(P),
\end{eqnarray}
where $\mathcal{S}(P)=-P\log(P)-(1-P)\log(1-P) $
%
%
%
%
is Shannon's entropy. For $\Psi_{A}(z_0(A))$ with $z_0(A)=Q/(1-Q)$, where $Q$ is the solution of the equation (\ref{eq:SP-large-N-Q}),  we obtain 
\begin{eqnarray}
\Psi_{A}(z_0(A))&=&-P\log(Q/(1-Q))-\log(1-Q)\nonumber\\
&&+\frac{1}{N}\sum_{\mu=1}^K\log\sum_{N^\alpha=0}^{\lfloor AN_\mu\rfloor} {N_\mu\choose N^\alpha}\,Q^{N^\alpha}[1-Q]^{N_\mu-N^\alpha}
\end{eqnarray}
and hence the difference
\begin{eqnarray}
\Psi_{A}(z_0(A))-\Psi_{1}(z_0(1))&=&-P\log(Q/(1-Q))-\log(1-Q)\nonumber\\
&&~~~~~~~+\frac{1}{N}\sum_{\mu=1}^K\log\sum_{N^\alpha=0}^{\lfloor AN_\mu\rfloor} {N_\mu\choose N^\alpha}\,Q^{N^\alpha}[1-Q]^{N_\mu-N^\alpha}-\mathcal{S}(P)\nonumber\\
&=&D(P\vert\vert Q)+\frac{1}{N}\sum_{\mu=1}^K\log\sum_{N^\alpha=0}^{\lfloor AN_\mu\rfloor} {N_\mu\choose N^\alpha}\,Q^{N^\alpha}[1-Q]^{N_\mu-N^\alpha}.
\end{eqnarray}

Let us consider now the derivative
\begin{eqnarray}
\ddot{\Psi}_{A}(z)&=&\frac{P}{z^2} +\frac{1}{N}\sum_{\mu=1}^K\frac{\partial^2}{\partial z^2}\log\phi_A(z\vert N_\mu)\nonumber\\
&=&\frac{P}{z^2} +\frac{1}{N}\sum_{\mu=1}^K\left[\frac{\ddot{\phi}_A(z\vert N_\mu)}{\phi_A(z\vert N_\mu)}-\left\{\frac{\dot{\phi}_A(z\vert N_\mu)}{\phi_A(z\vert N_\mu)}\right\}^2\right]\label{eq:d2Psi},
\end{eqnarray}
where 
\begin{eqnarray}
\dot{\phi}_{A}(z\vert N_\mu)=\sum_{N^\alpha=0}^{\lfloor AN_\mu\rfloor} {N_\mu\choose N^\alpha}z^{N^\alpha-1}N^\alpha\label{def:dphi}.
\end{eqnarray}
and
\begin{eqnarray}
\ddot{\phi}_{A}(z\vert N_\mu)=\sum_{N^\alpha=0}^{\lfloor AN_\mu\rfloor} {N_\mu\choose N^\alpha}z^{N^\alpha-2}N^\alpha(N^\alpha-1)\label{def:d2phi}.
\end{eqnarray}
For $z=Q/(1-Q)$, where $Q$ is a solution of the equation (\ref{eq:SP-large-N-Q}), above gives us 
\begin{eqnarray}
\ddot{\Psi}_{A}(z)&=&\left(\frac{1-Q}{Q}\right)^2\frac{1}{N}\sum_{\mu=1}^K\left[\langle(N^\alpha)^2\rangle_{A,Q,N_\mu;}-\langle N^\alpha\rangle^2_{A,Q,N_\mu;}\right]\label{eq:d2Psi-Q}
\end{eqnarray}
and for $z=P/(1-P)$ it gives us 
\begin{eqnarray}
\ddot{\Psi}_{1}(z)&=&\left(\frac{1-P}{P}\right)^2\frac{1}{N}\sum_{\mu=1}^K\left[\langle(N^\alpha)^2\rangle_{1,P,N_\mu;}-\langle N^\alpha\rangle^2_{1,P,N_\mu;}\right]\nonumber\\
&=&\left(\frac{1-P}{P}\right)^2\frac{1}{N}\sum_{\mu=1}^KN_\mu P(1-P)=\frac{\left(1-P\right)^3}{P}\label{eq:d2Psi_1}.
\end{eqnarray}
Hence using all of the above results in (\ref{eq:1-delta-SP-large-N}) gives us the equation (\ref{eq:SP-Prob-N-large-Q}).

\subsubsection{The union bound} 
The failure event (\ref{def:E}) is equivalent to the union   
$$\cup_{\mu=1}^K\left\{N_{\mu}^\alpha(\A)\geq \lfloor AN_{\mu}(\A)\rfloor+1\right\}$$ and hence the probability of failure $\delta=\Prob\left(\cup_{\mu=1}^K\left\{N_{\mu}^\alpha\geq \lfloor AN_{\mu}\rfloor+1\right\}\right)$. The latter  can be exploited to derive the following 
\begin{theorem}
\label{theorem:U-bound}
For random partitions  sampled from the probability distribution (\ref{def:Prob}) when  $\Prob(\alpha\vert\mu)<Q(\mu)<1$, where $Q(\mu)=\frac{\lfloor A N_\mu\rfloor+1}{N_\mu}$,  the probability of failure 

\begin{eqnarray}
%
\delta&\leq&\sum_{\mu=1}^K \rme^{-N\Phi(\Prob(\alpha, \mu),Q(\mu))}\label{eq:Prob-U-ub}, 
\end{eqnarray}
where 
\begin{eqnarray}
\Phi(\Prob(\alpha, \mu),Q(\mu))&=&-\log\left(\Prob(\mu)\,\rme^{-\mathrm{D}\left(Q(\mu)\vert\vert\Prob(\alpha\vert\mu)\right)}+1-\Prob(\mu)\right).\label{def:Phi}
\end{eqnarray}
\end{theorem}
\begin{proof}
First, by Boole's inequality, also known as  the union bound,  we have that
\begin{eqnarray}
\Prob\left(\cup_{\mu=1}^K\left\{N_{\mu}^\alpha\geq \lfloor AN_{\mu}(\A)\rfloor+1\right\}\right)
&\leq&\sum_{\mu=1}^K\Prob\left(N_{\mu}^\alpha\geq \lfloor AN_{\mu}\rfloor+1\right)\label{eq:Boole-ineq} 
\end{eqnarray}
and hence the probability of failure
\begin{eqnarray}
\delta &\leq&\sum_{\mu=1}^K\Prob\left(N_{\mu}^\alpha\geq \lfloor AN_{\mu}\rfloor+1\right)\label{eq:Boole-ineq-delta}. 
\end{eqnarray}
Second, using the result (\ref{eq:Prob-joint}) for the distribution (\ref{def:Prob}) we obtain the probability 
\begin{eqnarray}
\Prob\left(N^\alpha_\mu\geq \lfloor AN_{\mu}\rfloor+1\right)&=&\sum_{N_\mu=0}^N\sum_{N^\alpha_\mu=\lfloor AN_{\mu}\rfloor+1}^{N_\mu}\Prob^\alpha\!\left(N_\mu,N^\alpha_\mu\vert N\right)\nonumber\\
&=&\sum_{N_\mu=0}^N\Prob_\mu\!\left(N_\mu\vert N\right)\sum_{N^\alpha_\mu=\lfloor AN_{\mu}\rfloor+1}^{N_\mu}\Prob^\alpha_\mu\!\left(N^\alpha_\mu\vert N_\mu \right),
\end{eqnarray}
where in above we used that 
$\Prob^\alpha\!\left(N_\mu,N^\alpha_\mu\vert N\right)$, defined in  (\ref{eq:P-marginal}), is the marginal of (\ref{eq:Prob-joint}).

Now the probability 
\begin{eqnarray}
\Prob\left(N^\alpha_\mu\geq \lfloor AN_{\mu}\rfloor+1\vert N_{\mu}\right)&=&\sum_{N^\alpha_\mu=\lfloor AN_{\mu}\rfloor+1}^{N_\mu}\Prob^\alpha_\mu\!\left(N^\alpha_\mu\vert N_\mu \right)\leq \rme^{-N_\mu \mathrm{D}\left(Q(\mu)\vert\vert\Prob(\alpha\vert\mu)\right)}
\end{eqnarray}
by the \emph{Lemma} 4.7.2 in \cite{Ash2012} for $\Prob(\alpha\vert\mu)<Q(\mu)<1$, where $Q(\mu)=\frac{\lfloor AN_{\mu}\rfloor+1}{N_{\mu}}$, and hence we obtain 
\begin{eqnarray}
\Prob\left(N^\alpha_\mu\geq \lfloor AN_{\mu}\rfloor+1\right)&\leq&\sum_{N_\mu=0}^N\Prob_\mu\!\left(N_\mu\vert N\right)\rme^{-N_\mu \mathrm{D}\left(Q(\mu)\vert\vert\Prob(\alpha\vert\mu)\right)}\nonumber\\
&&=\left[\Prob(\mu)\,\rme^{-\mathrm{D}\left(Q(\mu)\vert\vert\Prob(\alpha\vert\mu)\right)}+1-\Prob(\mu)\right]^N\nonumber\\
&&=\rme^{-N\Phi(\Prob(\alpha, \mu),Q(\mu))},
\end{eqnarray}
which can be used to bound the sum in (\ref{eq:Boole-ineq-delta}) and hence giving us the inequality (\ref{eq:Prob-U-ub}).
\end{proof}
We note that a tighter, but slightly more complicated, upper bound on $\delta$ can be obtained if we use (\ref{eq:binom-tail-ub-2}) instead of (\ref{eq:binom-tail-ub-1}) in the proof. Furthermore, we obtained a less tight and slightly simpler upper bound in the following   
\begin{corollary}
 For random partitions generated  by the probability distribution (\ref{def:Prob}) when  $\Prob(\alpha\vert\mu)<Q(\mu)<1$, where $Q(\mu)=\frac{\lfloor A N_\mu\rfloor+1}{N_\mu}$,  the probability of failure
\begin{eqnarray}
\delta&\leq&\sum_{\mu=1}^K \rme^{-N\Phi(\Prob(\alpha, \mu),Q(\mu))}\leq\sum_{\mu=1}^K \rme^{-N\Prob(\mu)\,\phi(\Prob(\alpha\vert \mu),Q(\mu))}\label{eq:Prob-U-ub-2}, 
\end{eqnarray}
where $\phi(\Prob(\alpha\vert \mu),Q(\mu))=1-\rme^{-\mathrm{D}\left(Q(\mu)\vert\vert\Prob(\alpha\vert\mu)\right)}$.
\end{corollary}
\begin{proof}
Using  $\log(x)\leq x-1$ in the function (\ref{def:Phi}) gives us 
\begin{eqnarray}
-\Phi(\Prob(\alpha, \mu),Q(\mu))&\leq&\Prob(\mu)\left[\rme^{-\mathrm{D}\left(Q(\mu)\vert\vert\Prob(\alpha\vert\mu)\right)}-1\right]\nonumber\\
&&~~~~~~~~~~~~~~~~~~~~~~=-\Prob(\mu)\phi(\Prob(\alpha\vert \mu),Q(\mu))
\end{eqnarray}
and hence using above in (\ref{eq:Prob-U-ub}) gives us the second inequality in (\ref{eq:Prob-U-ub-2}).
\end{proof}
Finally,  a slight modification of the above proof also gives us the following   
\begin{corollary}
For random partitions generated  by the probability distribution (\ref{def:Prob-cond}) when  $\Prob(\alpha\vert\mu)<Q(\mu)<1$, where $Q(\mu)=\frac{\lfloor A N_\mu\rfloor+1}{N_\mu}$,  the probability of failure  
\begin{eqnarray}
\delta
&\leq&\sum_{\mu=1}^K \rme^{-N_\mu \mathrm{D}\left(Q(\mu)\vert\vert\Prob(\alpha\vert\mu)\right)}.\label{eq:Prob-U-ub-3}
\end{eqnarray}
\end{corollary}
\begin{proof}
To show this we exploit that the marginal of the probability distribution (\ref{eq:Prob-cond}), i.e. the result for  (\ref{def:Prob-cond}), is the binomial  (\ref{def:binomial-distr}). Thus if we replace $\Prob_\mu\!\left(N_\mu\vert N\right)$ in the proof of \emph{Theorem} \ref{theorem:U-bound} by the $\delta_{N_\mu;\tilde{N}_\mu}$, where $\sum_{\mu=1}^K\tilde{N}_\mu=N$, we obtain (\ref{eq:Prob-U-ub-3}). 
\end{proof}
We note that for the hypergeometric distribution (\ref{eq:Prob-cond-M})  the probability of failure
\begin{eqnarray}
\delta
&\leq&\sum_{\mu=1}^K \sum_{N^\alpha_\mu=\lfloor AN_{\mu}\rfloor+1}^{N_\mu}\Prob\left(N_{\mu}^\alpha\vert N_{\mu}; M\right)\label{eq:hyper-U-b-exact}\\
&&\leq \sum_{\mu=1}^K\rme^{-N_\mu \mathrm{D}\left(Q(\mu)\vert\vert M/N\right)}\label{eq:hyper-U-b}, 
\end{eqnarray}
where $M$ is the number of nodes of colour $\alpha$. The first inequality in above is an    application of (\ref{eq:Boole-ineq-delta}) with  the marginal (\ref{eq:hypergeom-univ})  and the second 
inequality is a consequence of  Hoeffding bound~\cite{Hoeffding1963}~\cite{Chvatal1979}. The union bounds (\ref{eq:Prob-U-ub}), (\ref{eq:Prob-U-ub-3}) and (\ref{eq:hyper-U-b}) suggest that for $K<\infty$ the probability of failure $\delta\rightarrow0$  if we assume that $N_\mu/N>0$ for all $\mu$ as $N\rightarrow\infty$.  However, it is easy to construct a partition for each of these bounds  such that the bound exceeds unity. Here it is interesting to compare the union bound in (\ref{eq:Prob-U-ub-3}) with the upper bound in (\ref{eq:lb-ub}) which was derived using a different approach.  The latter is bounded above by unity,  but  the former can be above unity  suggesting that  (\ref{eq:Prob-U-ub-3}) is a looser bound on $\delta$ which is  confirmed by the following
\begin{lemma}
For random partitions generated  by the probability distribution (\ref{def:Prob-cond}) when  $\Prob(\alpha\vert\mu)<Q(\mu)<1$, where $Q(\mu)=\frac{\lfloor A N_\mu\rfloor+1}{N_\mu}$,  the probability of failure  
\begin{eqnarray}
\delta&\leq& 1-\prod_{\mu=1}^K\left[1-\rme^{-N_\mu \mathrm{D}\left(Q(\mu)\vert\vert\Prob(\alpha\vert\mu)\right)}\right]
\leq\sum_{\mu=1}^K \rme^{-N_\mu \mathrm{D}\left(Q(\mu)\vert\vert\Prob(\alpha\vert\mu)\right)}.\label{eq:Prob-U-ub-4}
\end{eqnarray}
\end{lemma}
\begin{proof}
To show this we consider the difference 
\begin{eqnarray}
\Delta_K&=& 1-\prod_{\mu=1}^K\left[1-W(\mu)\right]
-\sum_{\mu=1}^K W(\mu),\label{eq:ub-diff}
\end{eqnarray}
where we defined the ``weight'' $W(\mu)=\rme^{-N_\mu \mathrm{D}\left(Q(\mu)\vert\vert\Prob(\alpha\vert\mu)\right)}$. The latter, because of $\mathrm{D}\left(Q(\mu)\vert\vert\Prob(\alpha\vert\mu)\right)\geq 0$, belongs to the interval  $(0,1]$ for any finite $N_\mu$.  
We note that $\Delta_1=0$ and $\Delta_2<0$. Let us assume that $\Delta_{K-1}\leq 0$ and consider $\Delta_{K}$ as follows 
\begin{eqnarray}
\Delta_K&=& 1-\left[1-W(K)\right]\prod_{\mu=1}^{K-1}\left[1-W(\mu)\right]
-\sum_{\mu=1}^{K-1} W(\mu)-W(K)\nonumber\\
&=&\Delta_{K-1}-  W(K)\left\{1-\prod_{\mu=1}^{K-1}\left[1-W(\mu)\right]
 \right\}.
\end{eqnarray}
Now from the $\prod_{\mu=1}^{K-1}\left[1-W(\mu)\right]\leq1$ and  $\Delta_{K-1}\leq0$  follows that $\Delta_{K}\leq0$ which, by induction on $K$, completes the proof. 
\end{proof}

\section{Application}
We consider $N$ nodes in the  blockchain network where a fraction of nodes, $P$,  is adversarial.  We assume that to improve its scalability, this network is split into $K$  shards (or committees). The latter allows the processing of all transactions, encoded  in the blockchain, in a distributed way. However, the security of the whole network can be only  guaranteed, with  \emph{probability} $1-\delta$, if $P$ is \emph{not} exceeding a certain threshold $A$ in \emph{each} committee.    The probability of failure $\delta$ considered in section  \ref{ssection:failure} corresponds to 
either the case when the number of adversarial nodes  is $N\times P$ on \emph{average}, where $\delta$ is given by the equation (\ref{eq:delta-binom}) with $\Prob(\alpha\vert\mu)=P$,  or to the case when it is \emph{exactly} $N\times P$, where $\delta$ is  given by the equation (\ref{eq:delta-hyper}) with $P=M/N$.  In the latter case, the probability $\delta$  was  computed, for  $N=nK$ nodes distributed into $K$ committees, only    by simulations~\cite{Hafid2020}, and the former case, to the best of our knowledge, has not been considered. Here we consider both cases for the number of nodes $N=nK+r$, where $n\in\mathbb{N}$ and $r\in\{0,\ldots,K-1\}$,  with $n$ and $n+1$  nodes distributed, respectively,   into  $K-r$ and   $r$  committees.
\begin{figure}[t]
\setlength{\unitlength}{0.57mm}
\begin{picture}(230,100)
\put(0,0){\includegraphics[height=100\unitlength,width=100\unitlength]{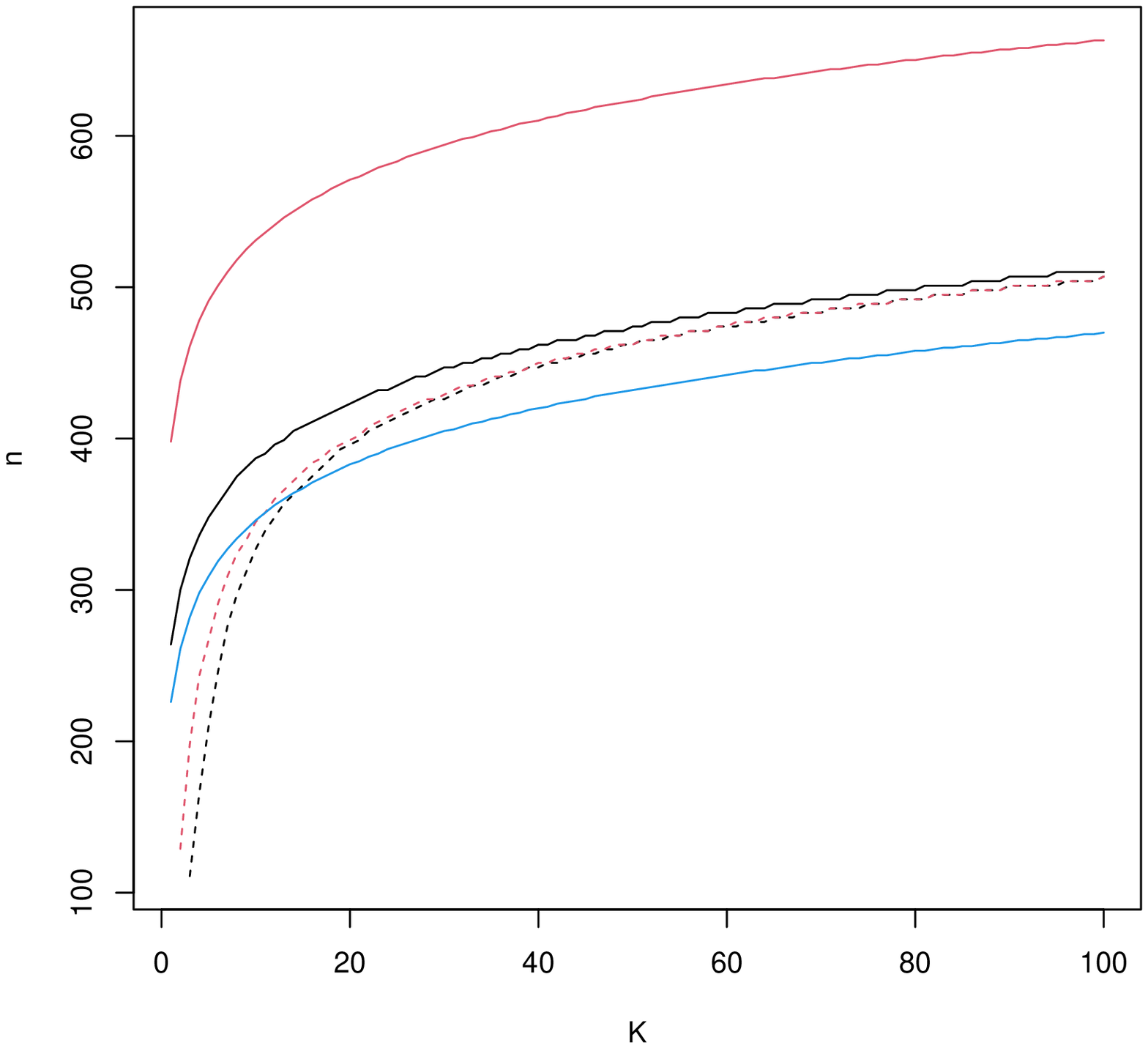}}
\put(115,-0){\includegraphics[height=100\unitlength,width=100\unitlength]{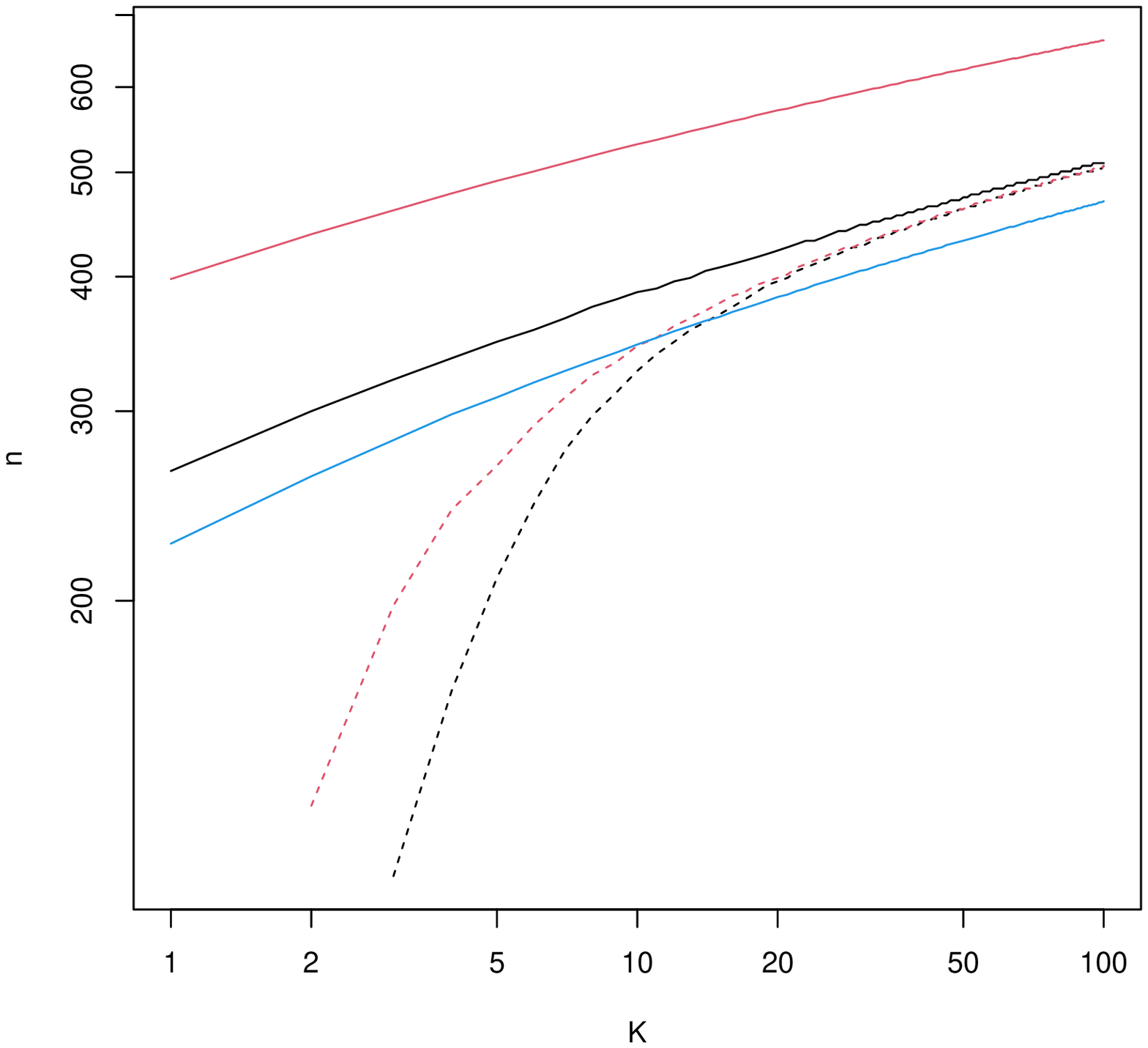}}
\end{picture}
\caption{The committee size $n$ as a function of the number of committees $K$ plotted for the probability of failure $\delta=10^{-3}$ when $A=1/3$ (fraction of a committee)  and $P=1/4$ (fraction of adversarial nodes). The  number of nodes $N=nK$ is increasing from left to right. The black solid  line is obtained by  solving  equation (\ref{eq:delta-binom}), which assumes $N\times P$ adversarial nodes on \emph{average}, numerically.  The red and blue solid lines  correspond to, respectively, the upper   and  lower bounds in (\ref{eq:lb-ub}).  The red and black dashed lines  were obtained from, respectively,  the upper bound (\ref{eq:hyper-U-b-exact}) and asymptotic result (\ref{eq:SP-Prob-N-large-Q})  for (\ref{eq:delta-hyper}).  The latter assumes \emph{exactly} $N\times P$ of adversarial nodes.} 
\label{fig:n-delta-1/1000}
\end{figure}

 \new{First}, we use the exact equation (\ref{eq:delta-binom}), the asymptotic result (\ref{eq:SP-Prob-N-large-Q})  for (\ref{eq:delta-hyper}), and the upper bound (\ref{eq:hyper-U-b-exact})  to find the committee size $n$ given the  probability of failure $\delta$ and number of committees $K$. We find that $n$ computed from the probability (\ref{eq:delta-binom}) is an upper bound on $n$ computed from the probability (\ref{eq:delta-hyper}) as can be seen in Figure \ref{fig:n-delta-1/1000}. The latter  is consistent with  our observation, that the probability (\ref{eq:delta-binom}) is an upper bound for (\ref{eq:delta-hyper}), \new{in simulations as can be seen } in the Figures \ref{fig:delta-N-1000} and \ref{fig:delta-N-10000}.  Furthermore,  for $\Prob(\alpha\vert\mu)=P$ and  $N=nK$, using  the bounds in (\ref{eq:lb-ub}), we obtain  
\begin{eqnarray}
1-\left[1-\frac{\rme^{-n \mathrm{D}\left(Q_n\vert\vert P\right)}}{\sqrt{8n\,Q_n\left(1-Q_n\right)}}\right]^{K}\leq\delta\leq 1-\left[1-\rme^{-n \mathrm{D}\left(Q_n\vert\vert P\right)}\right]^{K}\label{eq:lb-ub-n},
\end{eqnarray}
where we have defined $Q_n=\frac{\lfloor An\rfloor+1}{n}$.   The above inequalities  can be used to obtain bounds on the committee size $n$, given the probability of failure $\delta$ and  a number of committees $K$. The upper bound   
\begin{eqnarray}
n< \frac{-\log\left(1-\left(1-\delta\right)^{1/K}\right)}{\mathrm{D}\left(A\vert\vert P\right)} \label{eq:n-ub}
\end{eqnarray}
follows from the inequality 
\begin{eqnarray}
1-\left[1-\rme^{-n \mathrm{D}\left(Q_n\vert\vert P\right)}\right]^{K}
< 1-\left[1-\rme^{-n \mathrm{D}\left(A\vert\vert P\right)}\right]^{K}, 
\end{eqnarray}
where to obtain above we used that $A<Q_n\leq A+1/n$ and   that $\mathrm{D}\left(P+\epsilon\vert\vert P\right)$ is monotonic increasing function of 
$\epsilon$ when $2P-1<\epsilon\leq1-P$. Also the function  
\begin{eqnarray}
f(P+\epsilon)=
\mathrm{D}\left(P+\epsilon\vert\vert P\right)+\frac{1}{2n}\log\left((P+\epsilon)(1-P-\epsilon)\right)
%
\end{eqnarray}
is monotonic increasing in  
$\epsilon$ when  $2P-1<\epsilon<1/2-P$ and hence $f(Q_n)\leq f(A+1/n)\leq \tilde{f}(A)$, where $\tilde{f}(A)=\max_n f(A+1/n)$. Using the latter in the lower bound in (\ref{eq:lb-ub-n}) gives us 
\begin{eqnarray}
&&1-\left[1-\frac{\rme^{-n \tilde{f}(A)}}{\sqrt{8n}}\right]^{K}\leq1-\left[1-\frac{\rme^{-n \mathrm{D}\left(Q_n\vert\vert P\right)}}{\sqrt{8n\,Q_n\left(1-Q_n\right)}}\right]^{K}\leq\delta
\end{eqnarray}
and from above, using $\log(x)\geq 1-1/x$ for $x>0$, we obtain the lower bound
\begin{eqnarray}
 \frac{1-\log\left(8\right)-
 2\log\left(1-\left(1-\delta\right)^{1/K}\right)}{2\tilde{f}(A)+1}\leq n\label{eq:n-lb}.
\end{eqnarray}
Now for $K\rightarrow\infty$, with  $\delta\in(0,1)$, we have   
\begin{eqnarray}
 -\log\left(1-\left(1-\delta\right)^{1/K}\right)&=& -\log  \left( -\log  \left( 1-\delta \right)  \right) +\log  \left( K
 \right)\nonumber\\
 &&-{\frac {\log  \left( 1-\delta \right) }{2\,K}}-{\frac {
  \log^2\!  \left( 1-\delta \right) }{24\,{K}^{2}}}\nonumber\\
 &&+
\frac {  \log^4\!  \left( 1-\delta \right)}{2880\,{K}^{
4}}+O(1/K^6)\label{eq:large-K-expansion}
\end{eqnarray}
and for $\delta\rightarrow0$, with $K<\infty$, we have
\begin{eqnarray}
-\log\left(1-\left(1-\delta\right)^{1/K}\right)&=&\log  \left( K \right)+\log  \left( 1/\delta \right) -\frac {  K-1  
}{2\,K}\delta +O(\delta^2)\label{eq:small-delta-expansion}. 
\end{eqnarray}
Thus both the upper bound in (\ref{eq:n-ub}) and the lower bound in (\ref{eq:n-lb})  grow at most logarithmically in $K$  (or $1/\delta$) as  $K\rightarrow\infty$ (or as $\delta\rightarrow0$).   The latter suggests that to sustain the  same probability of failure $\delta$ (or the same number of committees $K$) the  committee size $n$ has to be increased   \emph{logarithmically} with $K$ (or $1/\delta$) as can be seen in Figure \ref{fig:n-delta-1/1000}.
\begin{figure}[t]
\setlength{\unitlength}{0.57mm}
\begin{picture}(230,100)
\put(0,0){\includegraphics[height=100\unitlength,width=100\unitlength]{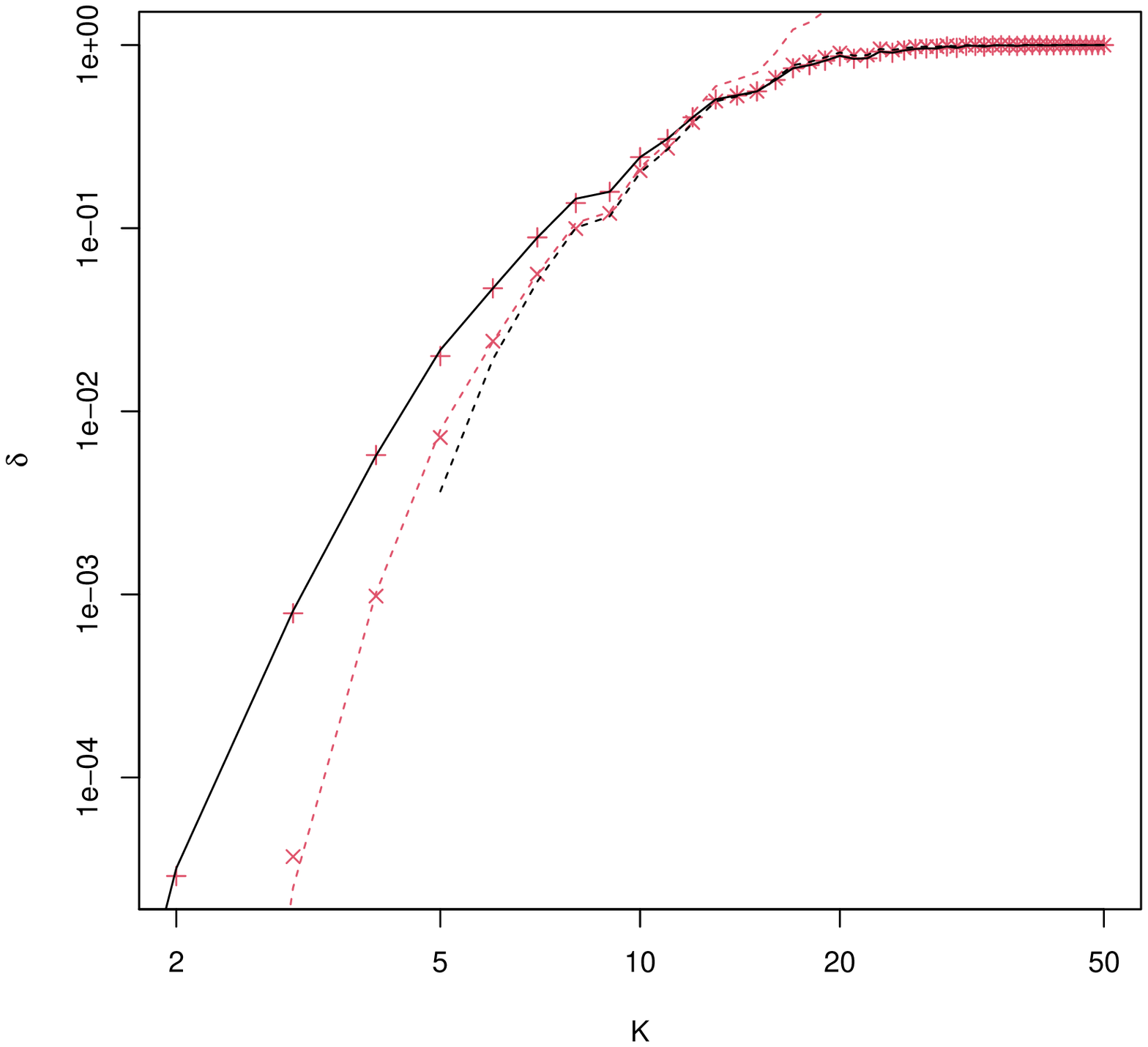}}
\put(115,-0){\includegraphics[height=100\unitlength,width=100\unitlength]{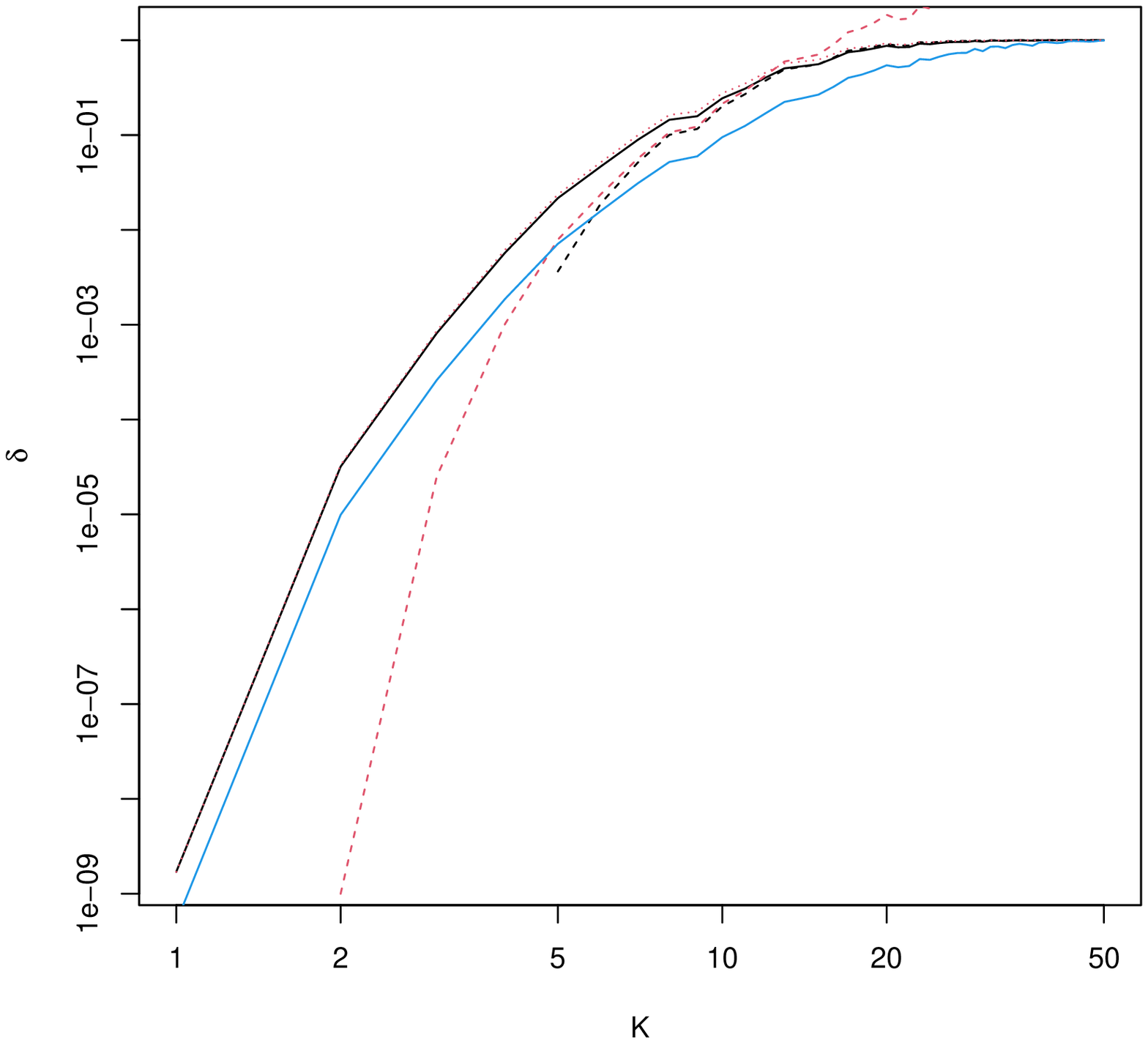}}
\end{picture}
\caption{The probability of failure, $\delta$,  as a function of the number of committees, $K$, computed for the parameters $A=1/3$ (fraction of a committee) and $P=1/4$ (fraction of adversarial or Byzantine nodes).   
Left: The red  $+$ and $\times$ symbols are results of simulations, obtained for $N=10^3$, where $N=nK+r$, nodes over   $10^6$ random samples, which assume the  $N\times P$  adversarial nodes, respectively, on \emph{average} and \emph{exactly}. 
The black solid line, going through the $+$ symbols,  connects exact values  computed numerically from  the equation (\ref{eq:delta-binom}). The black  dashed line corresponds to the asymptotic result (\ref{eq:SP-Prob-N-large-Q}) for (\ref{eq:delta-hyper}). The red dashed line  corresponds to  the  upper bound (\ref{eq:hyper-U-b-exact}). Right: The red dotted and blue solid lines correspond to, respectively, the upper and lower bounds in (\ref{eq:lb-ub}). The upper bound in the latter was computed using (\ref{eq:binom-tail-ub-2}).} 
\label{fig:delta-N-1000}
\end{figure}
\begin{figure}[t]
\setlength{\unitlength}{0.57mm}
\begin{picture}(230,100)
\put(0,0){\includegraphics[height=100\unitlength,width=100\unitlength]{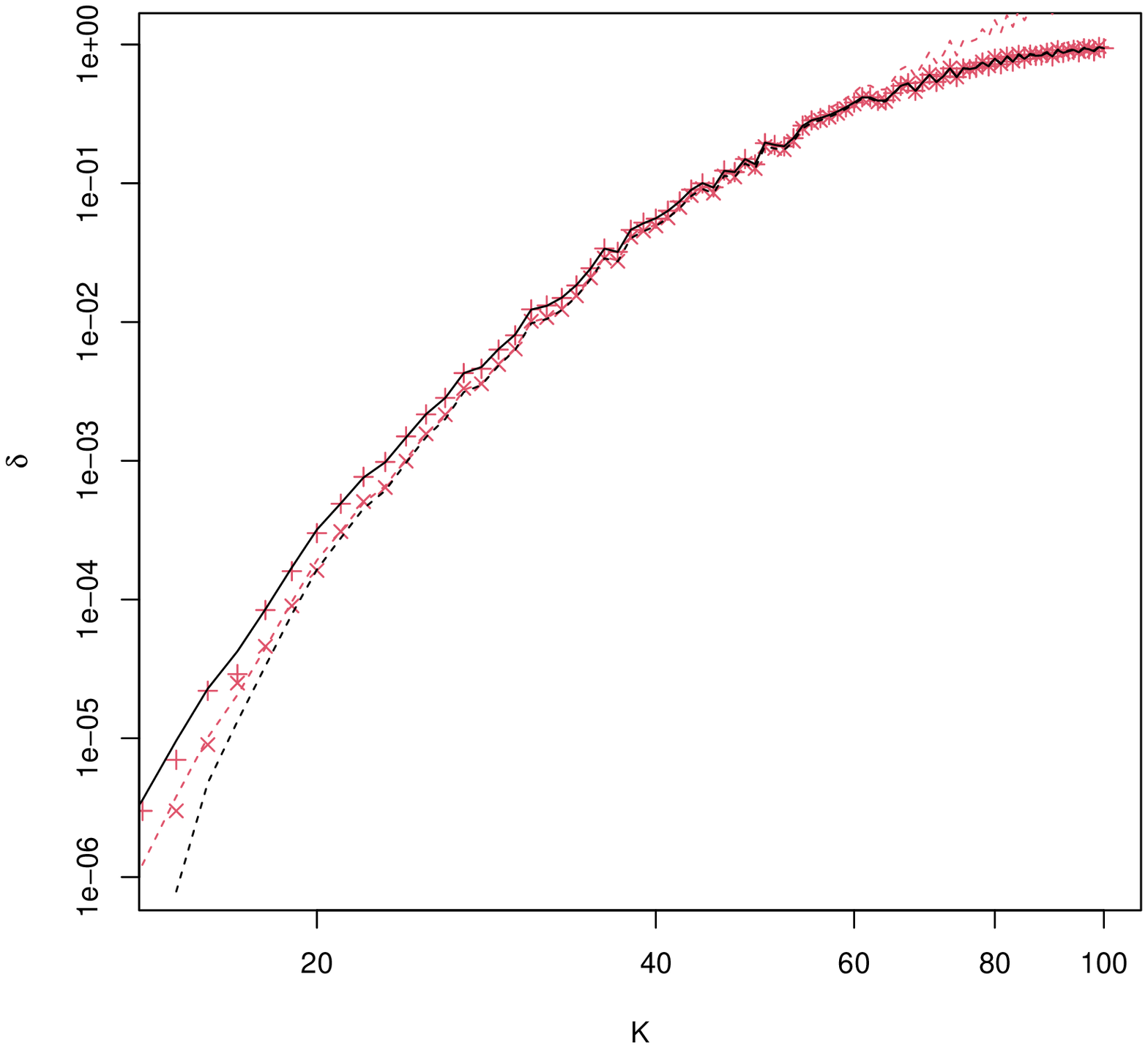}}
\put(115,-0){\includegraphics[height=100\unitlength,width=100\unitlength]{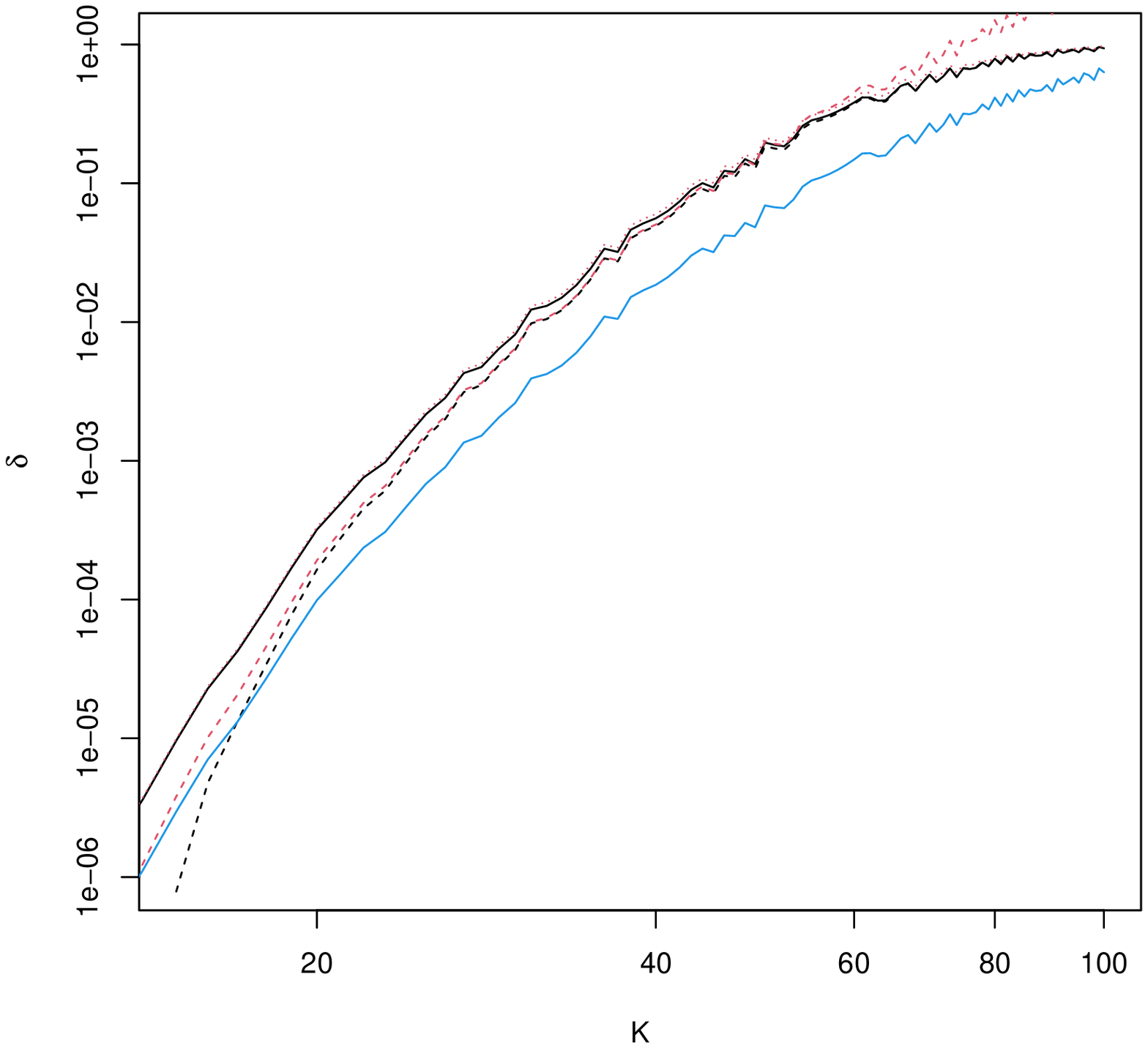}}
\end{picture}
\caption{The probability of failure, $\delta$,  as a function of the number of committees, $K$, computed for the parameters $A=1/3$ (fraction of a committee) and $P=1/4$ (fraction of adversarial or Byzantine nodes).   
Left: The red  $+$ and $\times$ symbols are results of simulations, obtained for $N=10^4$, where $N=nK+r$, nodes over   $10^6$ random samples, which assume the  $N\times P$  adversarial nodes, respectively, on \emph{average} and \emph{exactly}. 
The black solid line, going through the $+$ symbols,  connects exact values  computed numerically from  the equation (\ref{eq:delta-binom}). 
The black  dashed line corresponds to the asymptotic result (\ref{eq:SP-Prob-N-large-Q}) for (\ref{eq:delta-hyper}).The red dashed line  corresponds to  the  upper bound (\ref{eq:hyper-U-b-exact}). Right: 
The red  dotted  and blue solid lines correspond to, respectively, the upper and lower bounds in (\ref{eq:lb-ub}). The upper bound in the latter was computed using (\ref{eq:binom-tail-ub-2}). 
 } 
\label{fig:delta-N-10000}
\end{figure}

 \new{Second}, we test our analytic results  against  numerical experiments and we find, up to statistical variability and numerical accuracy, a good agreement  as can be seen in Figures \ref{fig:delta-N-1000} and \ref{fig:delta-N-10000}.   We also observe   that  the  probability   of failure (\ref{eq:delta-hyper})  is   bounded above by the probability (\ref{eq:delta-binom}). This is expected as  in the latter the number of adversarial nodes is a random number 
and in the former, this number is  fixed.  However the difference between the probabilities  (\ref{eq:delta-binom}) and (\ref{eq:delta-hyper})
is vanishing (see Figures \ref{fig:delta-N-1000}-\ref{fig:delta-N-10000-2}) when  both  probabilities approach unity with increasing $K$.

Finally, we note that insights drawn  from  the above results can be used to design an algorithm  which for a given number of nodes $N$  finds the \emph{maximum} the number of committees $K$ such that the probability of failure is less than (or equal) to  a given $\delta$. The latter will reduce the computational complexity in the  sharded blockchain, which is proportional to the size of a committee $n$, 
 without  compromising its safety~\cite{Hafid2020survey}. The possibly simplest way to implement such an algorithm is to use  the probability  of failure (\ref{eq:delta-binom}) and  $N=nK+r$ nodes with $n$ and $n+1$ nodes assigned, respectively, to the  $K-r$ and $r$ committees. The latter ensures that  for any $N$ the  optimization problem is essentially one-dimensional, where this  dimension is  $K$, and the former allows  the computation of a ``current'' probability of failure, which is  an elementary step  of  this optimization,  to be efficient.  We note that here we  assume that there are $N\times P$ adversarial nodes on average and using   this algorithm when there are  exactly $N\times P$ of adversarial   nodes  would give us  a larger committee size. However, the difference in committee sizes  between the latter and the former becomes  negligible   as $N$ is increasing as can be seen in  Figure \ref{fig:n-delta-1/1000}. The description and  pseudocode of one  possible variant for  such an algorithm are  provided in Appendix \ref{appendix:alg}.

\begin{figure}[t]
\setlength{\unitlength}{0.57mm}
\begin{picture}(230,100)
\put(0,0){\includegraphics[height=100\unitlength,width=100\unitlength]{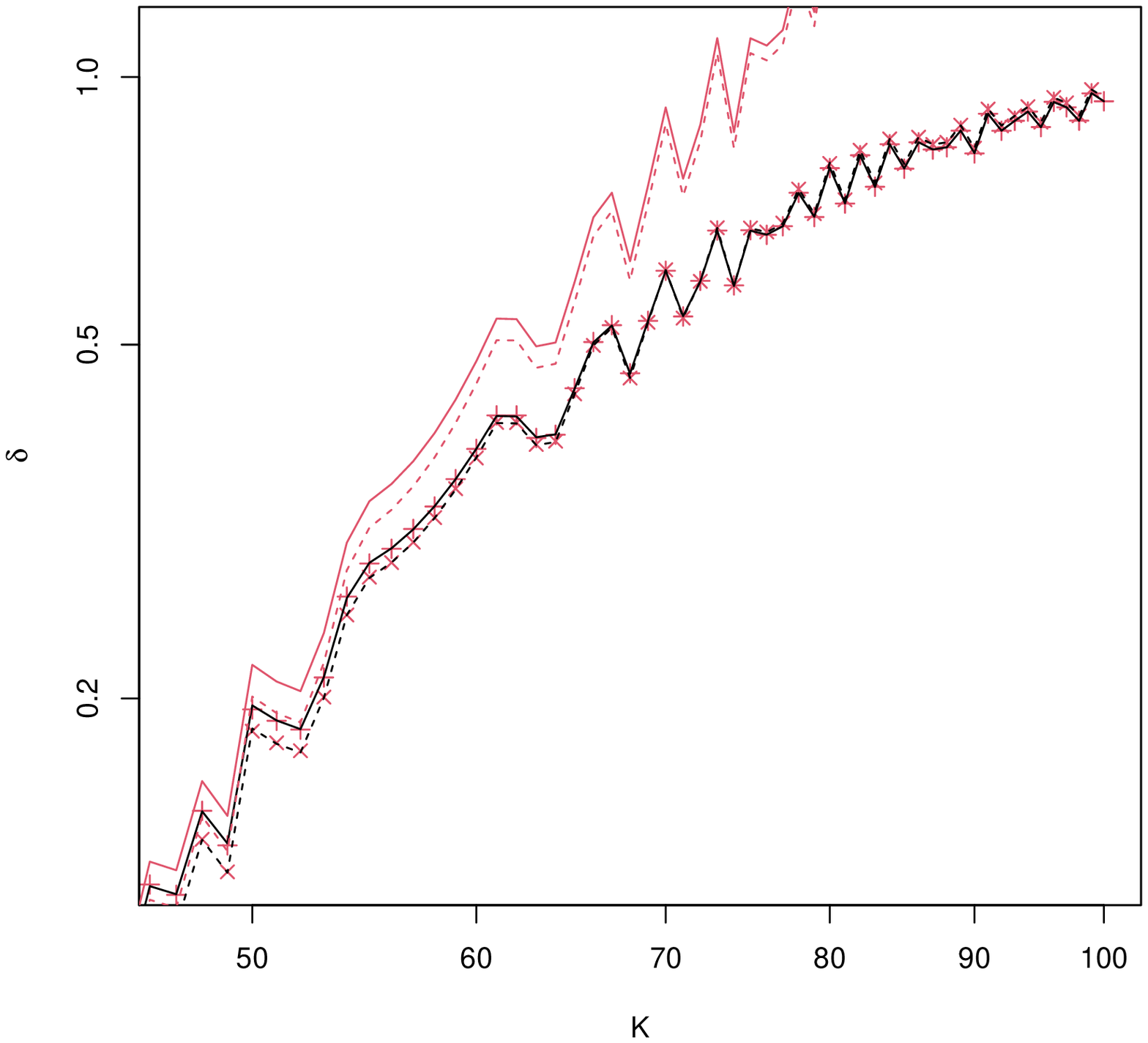}}
\put(115,-0){\includegraphics[height=100\unitlength,width=100\unitlength]{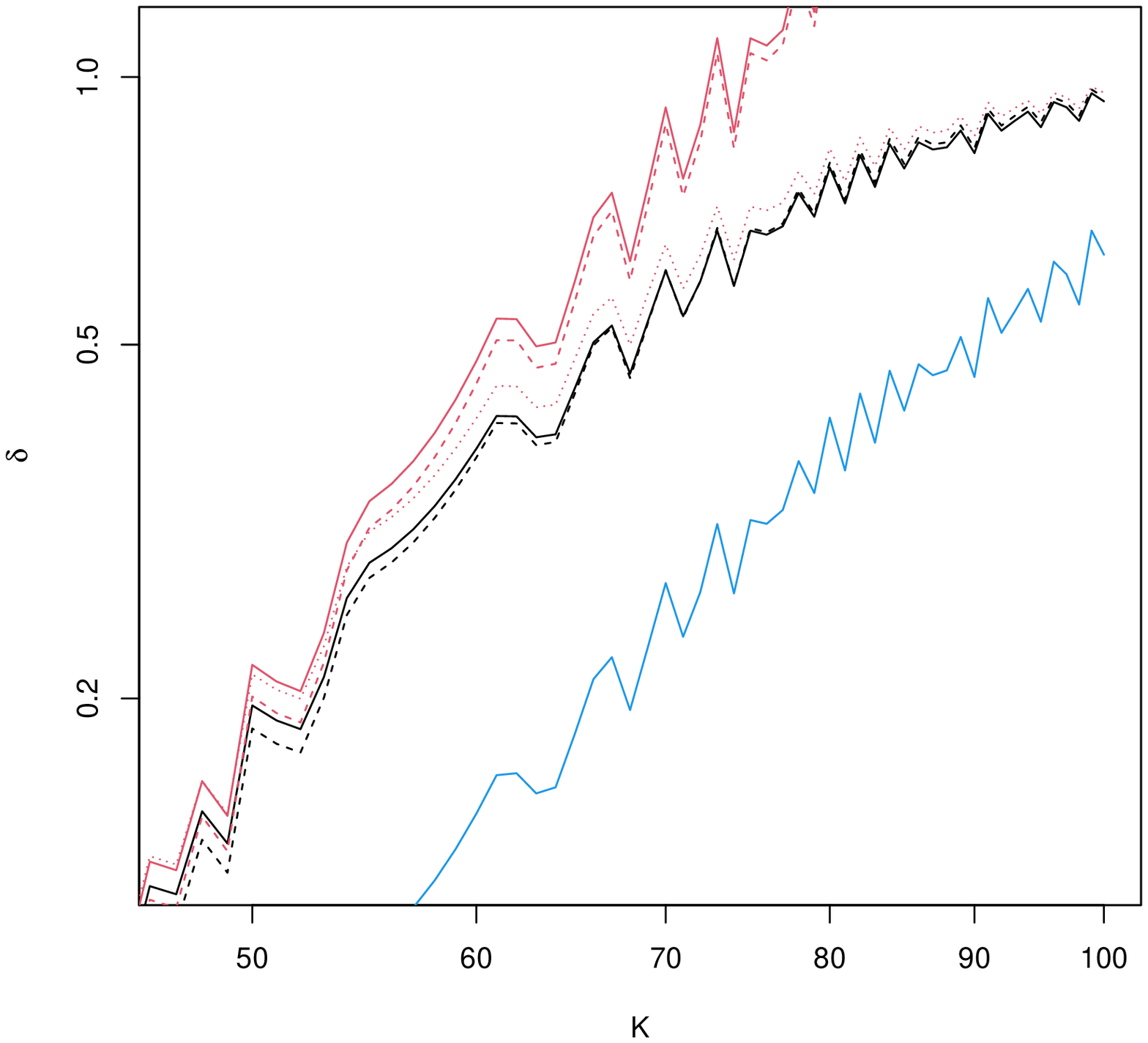}}
\end{picture}
\caption{Comparing analytic results of this work with typical  bounds from the  literature. Here the same data  as in the Figure \ref{fig:delta-N-10000} is used.
Left: The red  $+$ and $\times$ symbols are results of simulations
The black solid line, going through the $+$ symbols,  connects exact values  computed numerically from  the equation (\ref{eq:delta-binom}).
The black dashed line corresponds to the asymptotic result (\ref{eq:SP-Prob-N-large-Q}) for (\ref{eq:delta-hyper}). The red dashed line corresponds to  the union bound~\cite{Zamani2018}, i.e. the equation   (\ref{eq:hyper-U-b-exact}), which uses  hypergeometric distribution. The red solid line corresponds to the union bound which uses   binomial distribution~\cite{Kokoris2018} Right: The red  dotted  and blue solid lines correspond to, respectively, the upper and lower bounds in (\ref{eq:lb-ub}). The upper bound in the latter was computed using (\ref{eq:binom-tail-ub-2}).  } 
\label{fig:delta-N-10000-2}
\end{figure}
%
%

\section{Conclusion and Future work}
In this work, we have studied random partitions of networks. We introduced a very general probabilistic framework  and used this framework  to construct  probability distributions of  random partitions.  We showed that these distributions admit exact and explicit expressions which facilitated a study of the probability of failure. For the latter, we obtained exact expressions, bounds, and  asymptotic results.  For the  sharding of blockchains, these results offer significant  analytic  and algorithmic improvements,  but this work still leaves many questions open.

In particular,  we  have established  that  the probability of failure (\ref{eq:delta-binom}) is  an  upper bound on  the probability  (\ref{eq:delta-hyper}) but this  was done numerically and we envisage that to show this  in a more rigorous way is also possible. The other question is the relation between the asymptotic result (\ref{eq:SP-Prob-N-large-Q})  for the probability of failure (\ref{eq:delta-hyper}) and the upper bound (\ref{eq:hyper-U-b-exact}). Here, it is surprising that this bound is very accurate  when the   number of committees $K$  is small but for larger $K$  it also approaches the asymptotic result which has quite a different functional form. We note that the bound (\ref{eq:hyper-U-b-exact}) uses univariate hypergeometric distribution which can be difficult to compute numerically and hence for practical applications  a  simpler bound,  but  much tighter than (\ref{eq:hyper-U-b}), is more desirable.

\begin{acks}
The authors would like to thank the Nomos team for their invaluable assistance and very enlightening discussions that greatly enriched this work. 
\end{acks}


\appendix

\section{Description  of  the Algorithm\label{appendix:alg}}
Here we describe the algorithm which given the number of nodes $N$, assumed fraction of adversarial nodes  $P$, a fraction of a committee $A$ and the probability of failure $\delta$ computes the maximum number of committees $K$ and corresponding committee sizes. The algorithm computes  the RHS of  (\ref{eq:delta-binom}) for the  
   $N=nK+r$ nodes where  the  $n$ and $n+1$ nodes are assigned, respectively, to the  $K-r$ and $r$ committees. Initially, all $N$ nodes are in one committee, and in subsequent iterations, the number of committees $K$ is increased by one until the probability  (\ref{eq:delta-binom}) is equal to  or less than $\delta$. The pseudocode for the latter is provided below.  
   
\SetKwComment{Comment}{/* }{ */}
\RestyleAlgo{ruled}
\SetKwInOut{Input}{input}
\SetKwInOut{Output}{output}
\begin{algorithm}
\caption{The algorithm to compute minimal sizes of committees resilient to $\delta$ failure rate.}\label{alg:comm}
\Input{$N$, $\delta$, $A$, $P$;}
\Output{$K$, $n$, $r$, $\mathrm{Prob}$;}
$K \gets 1$\;
$n \gets N$\;
$r \gets 0$\;
$\mathrm{Prob} \gets 0$\;
\Repeat{$\mathrm{Prob}>\delta$}{
\tcc{Save $K$, $n$, $r$ and $\mathrm{Prob}$.}
$K_{-1}\gets K$\;
$n_{-1}\gets n$\;
$r_{-1}\gets r$\;
$Prob_{-1}\gets Prob$\;
\tcc{Compute next $K$.}
$K\gets K+1$\;
\tcc{Compute remainder, r, and quotient, n when N is divided by K.}
$r\gets rem(N,K)$\;
$n\gets quot(N,K)$\;
\eIf{$r>0$}{
  \tcc{Compute CDF of the Binomial(n,P) and Binomial(n+1,P).}
$\mathrm{Prob}_0\gets\Prob(X\leq \lfloor An\rfloor\,\vert n,P)$\; 
$\mathrm{Prob}_1\gets\Prob(X\leq \lfloor A(n+1)\rfloor\,\vert n+1,P)$\; 
\tcc{Compute the probability of failure.}
$Prob\gets 1-\mathrm{Prob}_0^{K-r}\mathrm{Prob}_1^r$  
  }{
  \tcc{Compute CDF of the Binomial(n,P).}
$\mathrm{Prob}_0\gets\Prob(X\leq \lfloor An\rfloor\,\vert n,P)$\;  
\tcc{Compute the probability of failure.}
$Prob\gets 1-\mathrm{Prob}_0^{K}$
    } 
}
$K \gets K_{-1}$\;
$n \gets n_{-1}$\;
$r \gets r_{-1}$\;
$Prob\gets \mathrm{Prob}_{-1}$\;
\end{algorithm}

\end{document}